\documentclass[envcountsame,envcountsect,runningheads, a4paper, 11pt]{llncs}

\usepackage[linesnumbered, noline]{algorithm2e}
\usepackage{amsmath}
\usepackage{amssymb}
\usepackage{caption}
\usepackage{etex}
\usepackage{tikz}
\usepackage{wrapfig}

\newcounter{prooffactcounter}
\newtheorem{prooffact}[prooffactcounter]{Fact}
\newcommand{\resetprooffacts}{\setcounter{prooffactcounter}{0}}
\spnewtheorem*{notation}{Notation}{\itshape}{}  

\usepackage{thm-restate}

\usetikzlibrary{positioning,arrows,automata}

\newcommand{\sendstep}[2]{\ensuremath{#1\rightarrow#2:}}
\newcommand{\range}{\ensuremath{\mathop{range}}}
\newcommand{\card}[1]{\left| #1 \right|}
\newcommand{\protocol}{\ensuremath{\tau}}
\newcommand{\prblemname}[1]{\ensuremath{\mathsf{#1}}}
\newcommand{\exstrat}{$\exists$-\prblemname{TStrat}}
\newcommand{\Kunde}{\ensuremath{\mathsf C}}
\newcommand{\Versand}{\ensuremath{\mathsf S}}
\newcommand{\Logistik}{\ensuremath{\mathsf D}}
\newcommand{\product}{\ensuremath{\mathsf{product}}}
\newcommand{\address}{\ensuremath{\mathsf{address}}}
\newcommand{\ptype}{\ensuremath{\mathsf{parceltype}}}
\newcommand{\delprice}{\ensuremath{\mathsf{deliveryprice}}}
\newcommand{\vars}{\ensuremath{\mathit{Vars}}}
\newcommand{\assign}{\ensuremath{\mathit{Assign}}}
\newcommand{\view}[1]{\ensuremath{\mathit{view}\left(#1\right)}}
\newcommand{\advcmd}[1]{\ensuremath{\mathsf{#1}}}
\newcommand{\set}[1]{\ensuremath\left\{#1\right\}}
\newcommand{\mathtext}[1]{\ensuremath{\mathrm{\text{#1}}}}
\newcommand{\var}[1]{\ensuremath{\vars(#1)}}
\newcommand\restr[2]{{
  \left.\kern-\nulldelimiterspace
  #1
  \vphantom{\big|}
  \right|_{#2}
  }}
\newcommand{\decisionproblemwidth}[4]{
\medskip
\vspace*{1mm}
\begin{tabular}{lp{#1}}
\textit{Problem:} & #2 \\
\textit{Input:} & #3 \\
\textit{Question:} & #4
\end{tabular}
\smallskip
\vspace*{1mm}
}
\let\doendproof\endproof
\renewcommand\endproof{~\hfill\qed\doendproof}
\newcommand{\keywords}[1]{\par\addvspace\baselineskip
\noindent\keywordname\enspace\ignorespaces#1}

\title{Active Linkability Attacks}

\author{Henning Schnoor \and Oliver Woizekowski}

\institute{Institut f\"ur Informatik, Christian-Albrechts-Universit\"at Kiel \\
       Olshausenstra\ss{}e 40, 24098 Kiel, Germany \\
       \email{$\{$henning.schnoor$|$oliver.woizekowski$\}$@email.uni-kiel.de}
       }

\begin{document}

\maketitle

\begin{abstract}
We study linking attacks on communication protocols. We show that an \emph{active} attacker is strictly more powerful in this setting than previously-considered passive attackers. We introduce a formal model to reason about active linkability attacks, formally define security against these attacks and give very general conditions for both security and insecurity of protocols. In addition, we introduce a composition-like technique that allows to obtain security proofs by only studying small components of a protocol.
\end{abstract}

\section{Introduction}

A typical goal of a protocol using web services is to compute values based on information that is distributed among several parties:
A user may have a specific set of input values, a web service then can compute, given these values, a function
whose result---possibly combined with further data supplied by the user---is then used as an input to
a further web service.

Such protocols can be synthesized to respect the privacy of individual values (e.g., address or credit card number)~\cite{BhargavanCorinFournet-SECURE-SESSIONS-FOR-WEB-SERVICES-ACMTISS-2007,BackesMaffeiPecinaReischuk-G2C-GOAL-DRIVEN-PROTOCOL-SPECIFICATIONS}. 

In addition to privacy of values, a crucial aspect in such a setting is \emph{linkability}: If an adversary can connect different values to the \emph{same} user session, this may be a privacy violation. For example, it might be harmless if an adversary knows customer names and products sold by a shop as long as these values cannot be \emph{linked}, i.e., the adversary does not learn who ordered what. Linkability has been studied in the context of eHealth protocols~\cite{DongJonkerPang-FORMAL-ANALYSIS-EHEALTH-PROTOCOL-ESORICS-2012} and anonymous internet usage~\cite{BiryukovPustogarovWeinmann-TORSCAN-ESORICS-2012}, similar privacy-relation questions have been considered in~\cite{NarayananShmatikov-DE-ANONYMIZATION-SP-2008,Sweeney-ACHIEVING-K-ANONYMITY-JUFKS-2002,ArapinisChothiaRitterRyan-UNLINKABILITY-AND-ANONYMITY-APPLIED-PI-CALCULUS-CSF_2010,EignerMaffei-DIFFERENTIAL-PRIVACY-BY-TYPING-CSF-2013}.

To the best of our knowledge, work on linkability up to now only studied what an adversary can deduce who does not interfere with the actual protocol run. Such attacks by an ``honest-but-curious'' adversary are \emph{passive}. We show that these attacks have an \emph{active} counterpart: An adversary involved in the actual protocol run is strictly more powerful. (We are not concerned with active attacks on the \emph{cryptographic} security of the protocols, there is of course a large literature on such attacks).

We contribute to the development of the theory of linkability by introducing active linking attacks. Our contributions are as follows:
\begin{itemize}
 \item We define a formal model that takes into account anonymous channels and nested web service queries.
 \item We give a formal definition of active linking attacks, and formalize a class of such attacks, which we call \emph{tracking strategies}.
 \item We demonstrate an \emph{embedding} technique which generalizes composition. This technique can be used to simplify security proofs.
 \item For a large, natural class of protocols, we give a complete characterization of secure protocols and possible attacks: Active linkability attacks can be mounted if and only if tracking strategies exist.
\end{itemize}

There are technical similarities between our security proofs and results on lossless decomposition of databases, where complete database tuples can be reproduced from partial ones~\cite{MaierMendelzonSagiv-IMPLICATIONS-DATA-DEPENDENCIES-TDS-1979,AhoBeeriUllman-JOINS-RELATIONAL-DATABASES-TDS-1979}. However, in the database setting, the notion of an active attacker is not studied.

\subsection{An Example}

\begin{wrapfigure}[7]{r}{5.4cm}
\vspace*{-5mm}
\begin{tabular}{ll}
 \sendstep \Kunde\Versand & $(\mathsf{product})$ \\
 \sendstep \Versand\Kunde & $p=\mathsf{parceltype}(\mathsf{product})$ \\
 \sendstep \Kunde\Logistik & $(p,\mathsf{address})$ \\
 \sendstep \Logistik\Kunde & $\mathsf{deliveryprice}(p,\mathsf{address})$
\end{tabular}
\caption{Simple Protocol $\protocol_{ex}$}
\label{fig:simples protokoll}
\end{wrapfigure}

A customer $\Kunde$ with $\address$ wants to learn the shipping cost for $\product$ ordered from shop $\Versand$ with shipping company service $\Logistik$. $\Kunde$ knows the values $\address$ and $\product$, $\Versand$ knows the function $\ptype$, determining the type of parcel $p=\ptype(\product)$ needed to package $\product$, $p$ is a number between $0$ and some $n$. The company $\Logistik$ knows the function $\delprice$ determining the shipping cost $\delprice(p,\address)$ of a parcel of type $p$ to $\address$. This setting yields the straight-forward protocol given in Figure~\ref{fig:simples protokoll}. (We abstract from cryptographic properties and assume secure channels between all parties.)

$\Kunde$ expects that $\Versand$ and $\Logistik$ cannot link $\product$ and $\address$, even if they work together: $\Versand$ learns $\product$ but not $\address$; $\Logistik$ learns $\address$ but not $\product$. If many users run the protocol in parallel and $\Kunde$ cannot be identified by his IP address (e.g., uses an anonymity service),
and $\Kunde$ waits a while between his two messages to avoid linking simply due to timing, then ideally $\Versand$ and $\Logistik$ should be unable to determine which of their respective queries come from the same customer. 

This reasoning is indeed correct for a passive attacker. However, it overlooks that $\Versand$ and $\Logistik$ control part of the user's data---namely the value $p$---and therefore can mount the following \emph{active} attack:

 \begin{enumerate}
  \item $\Versand$ replies to the first received query with a $0$, every later query is answered with a $1$. $\Versand$ stores the value of $\mathsf{product}$ from the first query.
  \item $\Logistik$ waits for a query of the form $(0, \address)$ and sends $\address$ to $\Versand$.
  \item $\Versand$ knows that $\address$ received from $\Logistik$ comes from the same user as the first query and hence can link this user's $\address$ and $\product$.
 \end{enumerate}
 
 This allows $\Versand$ and $\Logistik$ to produce a ``matching'' pair of $\address$ and $\product$, even with many parallel protocol runs and anonymous connections from $\Kunde$ to $\Versand$ and $\Logistik$. After one such run, the value $0$ can be used to track another session. Similarly, $n-1$ sessions can be tracked in parallel. The strategy can be refined in order to track a session in which a particular product was ordered. 
 The attack uses the $0$-value for $p$ as a ``session cookie'' that identifies a particular session. We stress that this attack does not violate a cryptographic property, but abuses control over user data to violate a privacy property of the protocol. In particular, this attack cannot be avoided by purely cryptographic means.

 \begin{wrapfigure}[9]{r}{4cm}
\vspace*{-0.8cm}
\begin{tikzpicture}
 \node at (0,0) (ptype) {\ptype};
 \node [below=10mm of ptype] (xprod) {$x_{\mathsf{product}}$};
 \node [right=5mm of xprod] (xaddress) {$x_{\mathsf{address}}$};
 \node at (1,1) (delprice) {\delprice};
 \draw [->] (xprod) edge (ptype);
 \draw [->] (ptype) edge (delprice);
 \draw [->] (xaddress) edge (delprice);
\end{tikzpicture}
\captionof{figure}{Model of $\protocol_{ex}$}\label{figure:into protocol as graph}
\end{wrapfigure}

This paper is organized as follows: In Section~\ref{section:model}, we introduce our protocol model and state our security definition. In Section~\ref{section:tracking strategies}, we generalize the above strategy to tracking strategies, which can be applied to a large class of protocols. In Section~\ref{section:security}, we present techniques to prove security of protocols, including a detailed proof for an example protocol, two general results for what ``flat'' protocols, a composition-like technique we call embedding, and generalizations of our security results to non-flat protocols. In Section~\ref{sect:generalized security}, we briefly discuss and characterize a generalization of our security notion. We then conclude, in Section~\ref{sect:conclusion}, with some ideas for further research.

\section{Protocol model}\label{section:model}

Our model provides anonymous channels between the user and each web service, since linking is trivial if the user can be identified by e.g., an IP address. For simplicity, we assume that all web services relevant for a protocol are controlled by a single adversary. To model interleaving of messages from different users, we introduce a scheduler who determines the order in which messages are delivered.

\begin{wrapfigure}[10]{l}{5.5cm}
\vspace*{-0.8cm}
\begin{tikzpicture}
 \node at   (0,0) (g) {$g$};
 \node [below=10mm of g] (f2) {$f_2$};
 \node [left= of f2] (f1) {$f_1$};
 \node [right= of f2] (f3) {$f_3$};
 
 \node [below=10mm of f1] (dummy) {};
 \node [left=3mm of dummy] (a1) {$a_1$};
 \node [right=3mm of a1] (a2) {$a_2$};
 \node [right=3mm of a2] (b1) {$b_1$};
 \node [right=3mm of b1] (b2) {$b_2$};
 \node [right=3mm of b2] (c1) {$c_1$};
 \node [right=3mm of c1] (c2) {$c_2$};
 
 \draw [->] (a1) edge [bend left] (f1);
 \draw [->] (a2) edge [bend right] (f1); 

 \node [left=-3mm of b1] (dummy2) {};
 \node [above=-1mm of dummy2] (dummy3) {};
 
 \draw [->] (dummy3) edge [bend right] (f1); 
 
 \node [right=-1mm of dummy3] (dummy5) {};
 
 \draw [->] (dummy5) edge [bend left] (f2);
 \draw [->] (b2) edge [bend right] (f2); 
 
 \node [left=-3mm of c1] (dummy6) {};
 \node [above=-1mm of dummy6] (dummy7) {};
 
 \draw [->] (dummy7) edge [bend right] (f2); 
 
 \node [right=-1mm of dummy7] (dummy8) {};
 
 \draw [->] (dummy8) edge [bend left] (f3);
 \draw [->] (c2) edge [bend right] (f3); 
 
 \draw [->] (a1) .. controls (2,-4) and (4,-3) ..  (f3);
 
 \draw [->] (f1) edge [bend left] (g);
 \draw [->] (f2) edge (g);
 \draw [->] (f3) edge [bend right] (g);
\end{tikzpicture}
\vspace*{-12mm}
\captionof{figure}{A protocol}\label{figure:protocol private variables}
\end{wrapfigure}
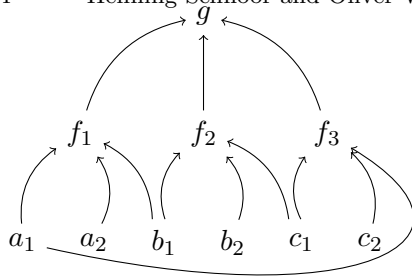

The security of a protocol depends on the structure of the nested queries to the involved web services. Since query results can be used as inputs for later queries, we model a protocol as a directed acyclic graph. Each node in this graph represents a query to a single web service such as $\mathsf{parceltype}$ and $\mathsf{deliveryprice}$ in the above example. Typically, these have some semantics describing the web service. However, we take the (pessimistic) point of view that the adversary ignores these semantics and replies only with the goal to maximize her attack chances. Therefore, our formal treatment does not fix any semantics for the functions computed in a protocol; we only distinguish between variable nodes (these model user input values) and query nodes (these model queries to web services).

An edge $u\rightarrow f$ in a protocol models that the value of $u$ (either an input value or a query result) is used as input to $f$. For simplicity, we assume that all values and query results in $\protocol$ are Boolean; other values can be modelled by introducing function domains or by encoding them as sequences of Booleans. User's input values are represented in the protocol using variables, these are the special nodes from $\var{\protocol}$. The representation of a protocol is similar to Boolean circuits (see~\cite{VollmerCircuitBook-SPRINGER-1999}).

\begin{definition}
 A \emph{protocol} is a directed acyclic graph $\protocol=(V,E)$ with a subset $\emptyset\neq\var{\protocol}\subseteq V$ such that each node in $\var{\protocol}$ has in-degree $0$.
\end{definition}

In Figure~\ref{figure:into protocol as graph}, the protocol $\protocol_{ex}$ from the introduction is formalized in our model, another example is presented in Figure~\ref{figure:protocol private variables}. 
Our protocols do not fix the order of requests to different services (except that if $f_1\rightarrow f_2$ is an edge in $\protocol$, each user must query $f_1$ before $f_2$). However, our results also hold for the case that the protocol fixes a query order. 

We call nodes of $\protocol$ without outgoing edges \emph{output nodes}. If $\protocol$ only has a single output-node, this node is the \emph{root} of $\protocol$.
We often identify $\protocol$ and its set of nodes, i.e., talk about nodes $f\in\protocol$ and subsets $\protocol'\subseteq\protocol$. For $u,v\in\protocol$, we write $u\rightsquigarrow v$ if there is a directed path from $u$ to $v$.

For $f\in\protocol$, $\vars(f)$ denotes $\set{x\in \vars(\protocol)\ \vert\ x\rightsquigarrow f\mathtext{ is a path in }\protocol}$, i.e., the set of input values that influence the queries made at the node $f$. For a set $S\subseteq\protocol$, with $\vars(S)$ we denote the set $\cup_{u\in S}\vars(u)$.

\subsection{Protocol execution}

We first informally describe how a protocol $\protocol$ is executed in our model. We identify a \emph{user} with her \emph{local session} containing her input values: A user or local session is an assignment $I\colon\var\protocol\to\set{0,1}$. During a protocol run, users store the results of queries. To model this, local sessions will be extended to assignments $I\colon\protocol\to\set{0,1}$. For a non-variable node $f\in\protocol$, the value $I(f)$ then contains the query result of $f$ for user $I$. $\assign$ is the set of all such assignments $I\colon V\to\set{0,1}$, where $V\subseteq\protocol$.

A \emph{run} or \emph{global session} of $\protocol$ is based on a multiset $S=\set{I_1,\dots,I_m}$ of users. Each $I_i$ performs a query for each non-variable node $f$ of $\protocol$ as follows: Let $u_1,\dots,u_n$ be the predecessor nodes of $f$ ($u_j$ represents a user's input value if $u_j\in\var\protocol$, and a result of a preceeding query otherwise). The arguments for the $f$-query are the user's values for $u_1,\dots,u_n$, i.e., the values $I_i(u_1),\dots,I_i(u_n)$. The query consists of the pair $(f,(I_i(u_1),\dots,I_i(u_n)))$. Hence the adversary learns which service the user queries and the arguments for this query, but does \emph{not} see the value $i$ identifying the user.

The adversary can reply to $I_i$'s $f$-query immediately or first wait for further queries. When she eventually replies with the bit $r$, the user stores this reply: We model this by extending $I_i$ with the value $I_i(f)=r$.

An adversary \emph{strategy} chooses one of three options in every situation:

\begin{enumerate}
 \item \emph{reply} to a previously-received user query,
 \item \emph{wait} for the next query (even if there are unanswered queries),
 \item \emph{print} an $I\in\assign$; the adversary wins if $I\in S$, and fails otherwise.
\end{enumerate}

\subsubsection{Schedules and Global Sessions}

Queries can be performed in any order that queries the predecessors of each node before the node itself. Formally, an \emph{$n$-user schedule for $\protocol$} is a sequence of pairs $(i,f)$ where $i\in\set{1,\dots,n}$, $f\in\protocol\setminus\vars(\protocol)$ where each such pair appears exactly once, and if $f\rightarrow g$ is an edge in $\protocol$ with $f\notin\var\protocol$, then $(i,f)$ appears in $\protocol$ before $(i,g)$. The pair $(i,f)$ represents the $f$-query of the $i$th user.

A \emph{global session} for $\protocol$ is a pair $(S,\sigma)$ where $S$ is a multiset of local sessions for $\protocol$, and $\sigma$ is a $\card{S}$-user schedule for $\protocol$.

\subsubsection{Protocol state.}

A \emph{protocol state} contains complete information about a protocol run so far. It is defined as a pair $(s,\sigma)$, where $s$ is a sequence over $\protocol \times \assign \times \mathbb N \times \set{0,1,\bot} \times (\mathbb{N} \cup \set{\bot})$ and $\sigma$ is a suffix of a schedule encoding the queries remaining to be performed. An element $s_i=(f,I,i,r,t)$ in $s$ encodes the $f$-query of user $I_i$ as described above, here $I$ is the assignment defined as $I(u)=I_i(u)$ for all $u\in\protocol$ where $u\rightarrow f$ is an edge in $\protocol$. The value $r$ is the adversary's reply to the query, $t$ records the time of the reply (both $r$ and $t$ are $\bot$ for a yet unanswered query).

The initial state of a global session $(S,\sigma)$ is $((f,\emptyset,1,\bot,\card{S}),\sigma)$ for some $f\in\protocol$; this initializes $\sigma$ and tells the adversary the number of users.

Two types of events modify the protocol state: A user can perform a query, and the adversary can reply to a query. (The adversary's print-action ends the protocol run).
The above-discussed query of a service $f\in\protocol$ by user $I_i$ can be performed in a global state ($s,\sigma)$ where the first element of $\sigma$ is $(i,f)$. This action adds the tuple $(f,I,i,\bot,\bot)$ to the sequence $s$, where $I$ encodes the input values for $f$ (see above), and removes the first element of $\sigma$.
In a state $(s,\sigma)$, if $s$ contains an element $s_k=(f,I,i,\bot,\bot)$ representing an unanswered query by $I_i$, the adversary's reply to this query with bit $r$ exchanges $s_k$ in $s$ with the tuple $(f,I,i,r,t)$, if this is the $t$-th action performed in the protocol run. Additionally, as discussed above, the assignment $I_i$ is then extended with $I_i(x_f)=r$. The remaining schedule is unchanged.

\subsubsection{Adversary knowledge and strategies.}

An adversary strategy is a conditional plan which, for each protocol state $(s,\sigma)$, chooses an adversary action to take. This action may only depend on information available to the adversary, which is defined by $\view{(s,\sigma)}$ obtained from $s$ by erasing each tuple's third component and ignoring $\sigma$. This models that the adversary has complete information except for the index of the user session from which a request originates and the remaining schedule. An adversary strategy for $\protocol$ is a function $\Pi$ whose inputs are elements $\view{(s,\sigma)}$ for a state $(s,\sigma)$ of $\protocol$, and the output is one of the above actions (\emph{wait}, \emph{reply} to element $s_k$ with $r$, \emph{print} assignment $I$), with the following restrictions:

\begin{itemize}
 \item \emph{reply} can only be chosen if $s$ contains an unanswered query $(f,I,i,\bot,\bot)$,
 \item \emph{wait} is only available if the first query in $\sigma$ can be performed, i.e., the first element of $\sigma$ is $(i,f)$ where $I_i(u)$ is defined for all $u$ with $u\rightarrow f$.\footnote{Whether \emph{wait} is available does not follow from $\view{(s,\sigma)}$. We can extend $\view{(s,\sigma)}$ with a flag for the availability of \emph{wait}, for simplification we omit this.}
\end{itemize}

For a global session and an adversary strategy, the resulting $\protocol$-run is defined as the resulting sequence of states arising from performing $\Pi$ stepwise, until the remaining schedule is empty and all queries have been answered, or the adversary's print action has been performed.

\begin{definition}\label{definition:secure protocol}
A protocol $\protocol$ is \emph{insecure} if there is a strategy $\Pi$ such that for every global session $(S,\sigma)$ of $\protocol$, an action $\advcmd{print}(I)$ for some $I \in S$ occurs during the $\protocol$-run for $(S,\sigma)$ with strategy $\Pi$. Otherwise, $\protocol$ is \emph{secure}.
\end{definition}

\section{Insecure Protocols: Tracking Strategies}\label{section:tracking strategies}

We now generalize the tracking strategy discussed in the introduction. That strategy used the value $0$ produced by \ptype\ as a ``session cookie'' to track the input values of a designated user session.
In our definition below, the node $t_{init}$ plays the role of $\ptype$ in the earlier example. At this node, tracking a user session is initialized by first replying with the session cookie: When the first $t_{init}$-query in a global session is performed, the adversary replies with the session cookie's value, $0$. The adversary stores the user's values used as arguments to $t_{init}$, this gives a partial assignment $I_{track}$ which is later extended by additional values: When the user later queries a node $f$ with $t_{init}\rightarrow f$ with the value $0$ for the argument representing $t_{init}$'s value, this query belongs to the tracked user session. The arguments for this $f$-query that contain the values of additional variables are then used to extend the assignment $I_{track}$, and the query is answered with the session cookie to allow tracking this session in the remainder of the protocol. In general, the $f$-query will not have further user values as input, but instead receive return values from different queries. However, since the adversary controls the replies to these queries as well, she can use them to simply ``forward'' the value of a user variable. If the values of all user variables can be forwarded to a node where tracking in the above sense happens, then by repeating these actions, the adversary eventually extends $I_{track}$ to a complete local session, which constitutes a successful active linkability attack. The following definition captures the protocols for which this attack is successful:

\begin{definition}
\label{definition:generalized tracking strategy}
A set $T \subseteq \protocol$ is called a \emph{tracking strategy} if the following conditions hold:
\begin{enumerate}
  \item Synchronization condition: There is a $\leadsto$-smallest element $t_{init}$ in $T$, i.e. for every $u \in T$ we have $t_{init} \leadsto u$.
  \item Cover condition: For every $x\in \vars(\protocol)$ there is a path $p_x$ such that:
    \begin{enumerate}
      \item $x \leadsto t$ via path $p_x$ for some $t \in T$
      \item if $x \neq y$, then $p_x$ and $p_y$ do not share a node from the set $\protocol \setminus T$
    \end{enumerate}
\end{enumerate}
\end{definition}

The set $T$ contains the nodes which perform tracking, i.e., which use the session cookie as reply to track the user session. 
The remaining nodes are used to simply forward one input value of the user to a later part in the protocol.
The cover condition guarantees that all input variables can be forwarded in this fashion.

The synchronization condition requires some node $t_{init}$ that can initialize tracking. The strategy then ensures that the session cookie used by each node in $T$ identifies the same session. This cookie is passed on to tracking nodes appearing later in the protocol run (indirectly via the users, who echo results of a query $f$ for all $g$ with $f\rightarrow g$). Without such a $t_{init}$ the adversary might use the session cookie for different partial user sessions that do not necessarily originate from the same local session. 

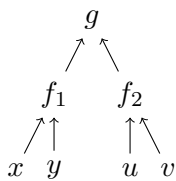
\begin{wrapfigure}[10]{l}{3.5cm}
\centering
\begin{tikzpicture}
  \node at (0,0) (g) {$g$};
  \node at (-0.5, -1) (f1) {$f_1$};
  \node at (0.5, -1) (f2) {$f_2$};
  \node at (-1.0,-2) (x) {$x$};
  \node at (-0.5,-2) (y) {$y$};
  \node at (0.5,-2) (u) {$u$};
  \node at (1.0,-2) (v) {$v$};

  \draw[->] (x) edge (f1);
  \draw[->] (y) edge (f1);
  \draw[->] (u) edge (f2);
  \draw[->] (v) edge (f2);
  \draw[->] (f1) edge (g);
  \draw[->] (f2) edge (g);
\end{tikzpicture}
\caption{No synchronization}
\label{figure:simple nosync example}
\end{wrapfigure}
Consider the example in Figure \ref{figure:simple nosync example}. We demonstrate that there is no tracking strategy for this protocol. In a tracking strategy,
both $f_1$ and $f_2$ must be tracking in order to capture all input bits. Both $f_1$ and $f_2$ then each collect a partial local session and use the session cookie $0$ to identify these. The adversary then waits for a query $g(0,0)$---however, such a query may never happen: If the initial queries at $f_1$ and $f_2$ belong to different user sessions, then two $g$-queries $g(1,0)$ and $g(0,1)$ will occur, which belong to different local sessions. We will show in Example~\ref{example:simple secure protocol} that this protocol does not only fail to have a tracking strategy, but is indeed secure: The adversary does not have \emph{any} strategy for a successful linkability attack.

\begin{wrapfigure}[11]{r}{4cm}
\centering
\vspace*{-5mm}
\begin{tikzpicture}
  \node at (0,0) (g) {$g$};
  \node at (-0.5, -1) (g1) {$f_1$};
  \node at (0.5, -1) (g2) {$f_2$};
  \node at (0,-2) (f2) {$t_{init}$};
  \node at (-1.5,-3) (w) {$w$};
  \node at (-0.5,-3) (x) {$x$};
  \node at (0.5,-3) (y) {$y$};
  \node at (1.5,-3) (z) {$z$};

  \draw[->] (w) edge (g1);
  \draw[->] (x) edge (f2);
  \draw[->] (y) edge (f2);
  \draw[->] (z) edge (g2);
  \draw[->] (f2) edge (g1);
  \draw[->] (f2) edge (g2);
  \draw[->] (g1) edge (g);
  \draw[->] (g2) edge (g);
\end{tikzpicture}
\caption{$t_{init}$ is a synchronizer}
\label{figure:sync example}
\end{wrapfigure}
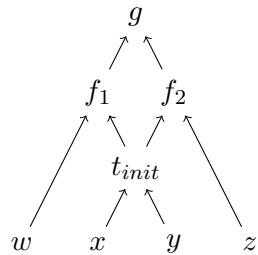
The situation is different 
for the protocol in Figure~\ref{figure:sync example}, where a \emph{synchronizer} $t_{init}$ is placed before $f_1$ and $f_2$: $t_{init}$ is the first node to see any data to be tracked, and generates the  cookie which is then passed on to $f_1$ and $f_2$; 
both $f_1$ and $f_2$ will at some point see a zero as its $t_{init}$-arguments. Then they contribute their knowledge to the local session $I_{track}$, since now they know they all work on the same local session.

Therefore, $T=\set{t_{init}, f_1, f_2}$ is a tracking strategy for the protocol in Figure~\ref{figure:sync example}. The synchronization condition is satisfied, since every node in $T$ can be reached from $t_{init}$ via a directed path, the cover condition is satisfied as well.

In this strategy, $t_{init}$ keeps track of the first encountered values of $x$ and $y$, which are stored in $I_{track}$. The node $f_1$ waits for a zero in its second argument and adds the received value for $w$ to $I_{track}$. Analogously, $f_2$ waits for a zero in its first argument and stores $z$. After both $f_1$ and $f_2$ have received this zero, $I_{track}$ is a complete assignment. The adversary thus can print the linked local session. Note that the root node is not involved in this process: The protocol remains insecure even without $g$.

A protocol for which a tracking strategy exists is always insecure:

\begin{restatable}{theorem}{theoremtrackingstrategieswork}\label{theorem:tracking strategies work}
 Let $\protocol$ be a protocol such that there exists a tracking strategy for $\protocol$. Then $\protocol$ is insecure.
\end{restatable}

We conjecture that the converse of Theorem~\ref{theorem:tracking strategies work} holds as well, i.e., that a protocol is insecure if and only if a tracking strategy exists. For a large class of protocols, we have proved this conjecture, see Theorem~\ref{theorem:each function has private variable}. Also, whether a tracking strategy exists for a protocol can be tested efficiently with a standard application of network flow algorithms.

\begin{wrapfigure}[5]{l}{2.2cm}
\vspace*{-8mm}
\begin{tabular}{l|l|ll|l}
      & $w$ & $x$ & $y$ & $z$ \\ \hline
$I_1$ & $1$ & $0$ & $0$ & $1$ \\
$I_2$ & $0$ & $0$ & $0$ & $0$
\end{tabular}
\captionof{figure}{$I_1$, $I_2$}\label{figure:global sessions inconsistency}
\end{wrapfigure}
We want to discuss one final point about the protocol from Figure~\ref{figure:sync example}, namely that the adversary's strategy here must be~\emph{inconsistent} in the following sense: Even when two local sessions agree on the values for $x$ and $y$, the strategy chooses diffeent replies to their $t_{init}$-queries. To see that this is necessary, consider a global session comprising only the two local sessions shown in Figure~\ref{figure:global sessions inconsistency} and a schedule as follows:
\begin{enumerate}
  \item First, $I_1$ queries $t_{init}$, which yields $I_1(t_{init})=0$ because $t_{init}$ is tracking. The adversary stores $I_{track}(x)=0$ and $I_{track}(y)=0$.
    \item Next, $I_2$ queries $t_{init}$, which gives $I_2(t_{init})=0$ because of consistency and $I_1$ and $I_2$ agree on $x$ and $y$.
    \item Lastly, $I_1$ queries $f_1$ with values $I_1(w)=1$ and $I_1(t_{init})=0$, which gives $0$ since $f_1$ is tracking and has not been queried before. The adversary stores $I_{track}(w)=1$. Analogously, $I_2$ queries $f_2$ with values $I_2(z)=0$ and $I_2(z)=0$, which also yields zero because of tracking. The adversary stores $I_{track}(z)=0$.  
\end{enumerate}
Now that both $f_1$ and $f_2$ have done their tracking, they combine their knowledge and print the local session. We see that consistent behavior in this case yields the output $\set{w=1, x=y=z=0}$, which is wrong.

\section{Secure Protocols and Security Proofs}\label{section:security}

We now present criteria implying security of protocols. We start with \emph{flat} protocols in Section~\ref{sect:two cases flat secure protocols}, for which we give a complete example security proof, and state two general security results. In Section~\ref{sect:deep:embedding}, we provide an embedding technique that allows to prove security of more complex protocols. We apply this technique in Section~\ref{sect:embedding application} to lift our results for flat protocols to protocols with arbitrary depth.

\subsection{Secure Flat Protocols: An Example and Two Results}\label{sect:two cases flat secure protocols}

A protocol is \emph{flat} if it has a root and its depth (i.e., length of longest directed path) is $2$. See Figure~\ref{figure:protocol private variables} for an example. A flat protocol can be written\footnote{One can without loss of generality assume that there is no variable $x$ and an edge $x\rightarrow g$ for the output node $g$ of a flat protocol.} as $\protocol=g(f_1(\overrightarrow{x_1}),\dots,f_n(\overrightarrow{x_n}))$, where $\overrightarrow{x_i}$ is a sequence of variables ($\overrightarrow{x_i}$ and $\overrightarrow{x_j}$ are not necessarily disjoint). For example, the protocol from Figure~\ref{figure:protocol private variables} can be written as $g(f_1(a_1,a_2,b_1),f_2(b_1,b_2,c_1),f_3(c_1,c_2,a_1))$.

We now present an example and two classes of secure flat protocols.

\begin{example}\label{example:simple secure protocol}
 The protocol $\protocol=g(f_1(x,y),f_2(u,v))$ (see Fig.~\ref{figure:simple nosync example}) is secure.
\end{example}

\begin{wrapfigure}[9]{l}{6.5cm}
 \begin{tabular}{c|cc|cc|c}
        & $x$ & $y$ & $u$ & $v$ & $g$-query \\ \hline
  $I_1$ & $0$ & $0$ & $0$ & $0$ & $(0,\gamma)$ \\
  $I_2$ & $0$ & $1$ & $0$ & $1$ & $(0,\delta)$ \\
  $I_3$ & $1$ & $0$ & $1$ & $0$ & $(\alpha,0)$ \\
  $I_4$ & $1$ & $1$ & $1$ & $1$ & $(\beta,0)$
 \end{tabular}\ \ \ 
 \begin{tabular}{c|cc|cc|c}
        & $x$ & $y$ & $u$ & $v$ & $g$-query \\ \hline
  $I_1$ & $0$ & $0$ & $0$ & $1$ & $(0,\delta)$ \\
  $I_2$ & $0$ & $1$ & $0$ & $0$ & $(0,\gamma)$ \\
  $I_3$ & $1$ & $0$ & $1$ & $1$ & $(\alpha,0)$ \\
  $I_4$ & $1$ & $1$ & $1$ & $0$ & $(\beta,0)$
 \end{tabular}
 \captionof{figure}{Two Global Sessions for $\protocol$}\label{figure:global sessions}
\end{wrapfigure}
As discussed above, there is no tracking strategy in the sense of Definition~\ref{definition:generalized tracking strategy} for $\protocol$. We now show that $\protocol$ is indeed secure, i.e., there is no strategy at all for $\protocol$.

\begin{proof}
We only consider global sessions consisting of $4$ different local sessions $I_1,\dots,I_4$, where each Boolean combination appears as input to $f_1$ and $f_2$: For each $\alpha,\beta\in\set{0,1}$, there is some $I_i$ with $I_i(x)=\alpha$ and $I_i(y)=\beta$, and a session $I_j$ with $I_j(u)=\alpha$ and $I_j(v)=\beta$. We only consider schedules $\sigma$ that first perform all $f_1$-queries, then all $f_2$-queries followed by all $g$-queries, and perform the queries for each service in lexicographical order. It suffices to show that the adversary does not have a strategy for this case, clearly then a general strategy does not exist either. In sessions as above, the adversary always receives the same set of queries for $f_1$ and $f_2$; therefore these queries contain no information for the adversary. Since no $f_i$ is queried more than once with the same arguments, we only need to consider consistent strategies. As a result, each adversary strategy $\Pi$ consists of functions $f_1,f_2\colon\set{0,1}^2\rightarrow\set{0,1}$ and a rule for the print action. We construct, depending on $f_1$ and $f_2$, a global session where $\Pi$ fails. Hence let $f_1$ and $f_2$ be functions as above.

Since $f_1$ and $f_2$ cannot be injective, we assume without loss of generality that $f_1(0,0)=f_1(0,1)=0$ and $f_2(1,0)=f_2(1,1)=0$. We further define $\alpha=f_1(1,0)$, $\beta=f_1(1,1)$, $\gamma=f_2(0,0)$, and $\delta=f_2(0,1)$. 

Now consider the global sessions in Figure~\ref{figure:global sessions}, with a schedule as above. For local sessions appearing, the tables list the values of the protocol variables as well as the parameters for the resulting query at the node $g$, which for the local session $I_i$ consists of the pair $(f_1(I_i(x),I_i(y)),f_2(I_i(u),I_i(v)))$. These $g$-queries are observed by the adversary. It turns out that for both global sessions, the adversary makes the exact same observations: At both $f_1$ and $f_2$, each Boolean pair is queried exactly once; the queries at $g$ are $(0,\gamma)$, $(0,\sigma)$, $(\alpha,0)$, and $(\beta,0)$---all performed in lexicographical order. Therefore, the adversary cannot distinguish these sessions and his strategy prints out the same assignment in both global sessions. Since the sessions have disjoint sets of local sessions, the adversary fails in at least one of them;  the protocol is indeed secure.
\end{proof}

The ideas from the above proof can be generalized to give the following theorem. Its proof is technically more involved in part due to the fact that here, we have to consider inconsistent adversary strategies. 
Roughly, the theorem states that if a flat protocol can be partitioned into two variable-disjoint components, neither of which grant the adversary enough ``channels'' to forward all user inputs to the output node, then it is secure.

\begin{restatable}{theorem}{theoremdisjointvariables}\label{theorem:disjoint variables}
 Let $\protocol$ be a protocol of the form $\protocol=g(f_1(\overrightarrow{x_1}),\dots,f_n(\overrightarrow{x_n}))$, such that $\set{1,\dots,n}=I_1\cup I_2$ with 
 \begin{itemize}
  \item if $i\in I_1$ and $j\in I_2$, then $\overrightarrow{x_i}\cap\overrightarrow{x_j}=\emptyset$,
  \item $\card{\cup_{i\in I_j}\overrightarrow{x_i}}>\card{I_j}\ge1$ for $j\in\set{1,2}$.
 \end{itemize}
 Then $\protocol$ is secure.
\end{restatable}

We now consider flat protocols $\protocol=g(f_1(\overrightarrow{x_1}),\dots,f_n(\overrightarrow{x_n}))$ where each $f_i$ has a private variable. This is a variable $x\in\overrightarrow{x_i}\setminus\cup_{j\neq i}\overrightarrow{x_j}$, i.e., a variable that is an input to $f_i$, but not an input to any of the other $f_j$. For these protocols we show that the converse of Theorem~\ref{theorem:tracking strategies work} is true as well: If there is no tracking strategy for $\protocol$, then the protocol is secure. 

\begin{restatable}{theorem}{theoremeachfunctionhasprivatevariable}\label{theorem:each function has private variable}
 Let $\protocol$ be a flat protocol where each $f_i$ has a private variable. Then $\protocol$ is insecure if and only if a tracking strategy for $\protocol$ exists.
\end{restatable}

The proof of Theorem~\ref{theorem:each function has private variable} converts $\protocol$ to a normal form, and applies a rather involved combinatorial construction to construct a global session $(S,\sigma)$ such that for every local session $I\in S$, there is a global session $(S_I,\sigma_I)$ which is indistinguishable from $(S,\sigma)$ for the adversary but does not contain $I$. Due to this indistinguishability, each adversary strategy has to print the same local session on $(S,\sigma)$ and each $(S_I,\sigma_I)$, and hence fails on $(S,\sigma)$ or on some $(S_I,\sigma_I)$. The simplest non-trivial example for which Theorem~\ref{theorem:each function has private variable} implies security is the protocol shown in Figure~\ref{figure:protocol private variables}, which itself already requires a surprisingly complex security proof.

\subsection{A Notion of Composition}\label{sect:deep:embedding}

We introduce an \emph{embedding} technique to compare security of protocols: If a copy of a secure protocol $\protocol'$ appears as a (loosely speaking) ``component'' of a protocol $\protocol$, then $\protocol$ is secure as well---provided that the ``copy'' of $\protocol'$ plays a meaningful part in $\protocol$: The copy must have control over its input data, and be applied to ``relevant'' input values.

This notion is interesting for several reasons. First, it is a powerful tool to prove security of protocols into which a known secure protocol can be embedded (see Section~\ref{sect:embedding application} for such applications). Second, embedding is closely related to \emph{composition}, a technique establishing implications between the security of a protocol and its components. Various notions of composition have been successfully applied in the study of cryptographic protocols~\cite{Canetti-UCMODEL-FOCS-2001,BackesPfitzmannWaidner-CryptographicLibrary-CCS-2003,CiobacaCortiar-PROTOCOL-COMPOSITION-ARBITRARY-PRIMITIVES-CSF-2010,ChevalierDelauneKremerRyan-COMPOSITION-PASSWORD-BASED-PROTOCOLS-FMSD-2013}. Usually, composition uses the output of one protocol as input for others. In a sense, embedding is more general, as it also captures the case where (copies of) the nodes of the embedded protocol $\protocol'$ appear spread out over different parts of the ``large'' protocol $\protocol$.

Our embedding notion only requires \emph{one} component of the composed protocol to be secure, whereas composition techniques usually compose a number of secure protocols to obtain another, equally secure, protocol. The reason for this difference lies in our notion of security: In a linkability attack, the adversary needs to reconstruct a \emph{complete} user session to be successful. Clearly, if a subset of the user's variables is protected by a secure sub-protocol, then even complete knowledge of all other values does not help the adversary. In contrast, definitions of cryptographic security in the literature usually require that the adversary cannot attack any part of the protocol, in which case a single insecure sub-protocol usually renders the entire protocol insecure. See Section~\ref{sect:generalized security} for a relaxation of the requirement that the adversary must construct a complete user session.

We now state our embedding definition and then discuss it in detail. For a function $\varphi\colon\protocol'\to\protocol$, with $\varphi(\protocol')$ we denote the set $\set{\varphi(u)\ \vert\ u\in\protocol'}$.

\begin{definition}
 Let $\protocol$ and $\protocol'$ be protocols. A function $\varphi\colon\protocol'\rightarrow\protocol$ is an \emph{embedding} of $\protocol'$ into $\protocol$, if the following holds:
 \begin{itemize}
  \item If $v\in\protocol'$ is a node of $\protocol'$ with a successor in $\protocol'$ and $\varphi(v)\rightsquigarrow u$ is a path in $\protocol$, then there is some $w\in\protocol'$ such that $u\rightsquigarrow\varphi(w)$ is a path in $\protocol$.
  \item If $u\neq v\in\protocol'$ with a path $\varphi(u)\rightsquigarrow\varphi(v)$ in $\protocol$ whose intermediate nodes are not elements of $\varphi(\protocol')$, then there is an edge $u\rightarrow v$ in $\protocol'$.
  \item there is an injective function $\chi\colon\vars(\protocol')\rightarrow\vars(\protocol)$ such that for each $x\in\vars(\protocol')$, there is a path $\chi(x)\rightsquigarrow\varphi(x)$ in $\protocol$, and if $x,y\in\var{\protocol'}$ with $\chi(x)\rightsquigarrow w$, $\chi(y)\rightsquigarrow w$ are paths in $\protocol$, then there is a node $g\in\protocol'\setminus\var{\protocol'}$ such that both of these paths visit $\varphi(g)$.
 \end{itemize}
\end{definition}

The definition requires $\varphi(\protocol')$, which is the copy of $\protocol'$ appearing in $\protocol$, to have sufficient control a set of revelant user input values:
\begin{enumerate}
 \item Part one of the definition states that when information ultimately leaves $\varphi(\protocol')$, it does so through the output interface of $\protocol'$. Hence $\varphi(\protocol')$ retains some control over its data, even though data may be processed by nodes from $\protocol\setminus\varphi(\protocol')$ in the meantime. This prevents ``internal'' data of $\varphi(\protocol')$ from being copied to $\protocol\setminus\varphi(\protocol')$ without any control by $\varphi(\protocol')$.
 \item The second point requires that in $\varphi(\protocol')$, there are no connections that do not appear in the original protocol $\protocol'$. This is needed since e.g., the transitive closure of almost every protocol is insecure. 
 \item The third part demands that $\varphi(\protocol')$ is used on ``exclusive'' data: Each variable $x$ of  $\protocol'$ must correspond to some $\chi(x)\in\var{\protocol}$, such that $\chi(x)$ and $\chi(y)$ can be linked only through their interaction within $\varphi(\protocol')$. Otherwise, security of $\protocol'$ cannot prevent linking $\chi(x)$ and $\chi(y)$.
\end{enumerate}

Embedding preserves security: A protocol with a secure component is secure itself. Hence, to prove a protocol secure, it suffices to find a secure component which then forms an obstruction to any adversary strategy.

\begin{restatable}{theorem}{theoremembedding}\label{theorem:embedding}
  Let $\protocol'$ be a secure protocol with a root, and let $\varphi$ be an embedding of $\protocol'$ into a protocol $\protocol$. Then $\protocol$ is secure as well.
\end{restatable}

The proof of Theorem~\ref{theorem:embedding} shows that $\protocol'$ can be obtained from $\protocol$ by a sequence of insecurity-preserving transformations. Hence, if $\protocol$ is insecure, then $\protocol'$ is insecure as well. Technically, the proof uses fairly natural translations of strategies between different intermediate protocols, however we have to be careful to account for inconsistent strategies.

As an example, we discuss one of these transformations, which we call \emph{cloning}. This discussion also allows to compare security of two different protocol approaches.

Cloning is used in the proof of Theorem~\ref{theorem:embedding} in the case that $\varphi(u)=\varphi(v)=w$ for some nodes $u$ and $v$ of $\protocol'$ with $u\neq v$. Since the goal of the transformation is to obtain \emph{exactly} the protocol $\protocol'$, the nodes $\varphi(u)$ and $\varphi(v)$ have to be ``separated'' in $\protocol$. To this end, cloning introduces a new ``copy'' $w'$ of $w$ with the same predecessors as $w$ into $\protocol$, and replaces some edges $w\rightarrow t$ with edges $w'\rightarrow t$. The result is a protocol where instead of caching the $w$-result and reusing the value as input for later queries, the user performs the query to $w$ again. Clearly, this is not a good idea as the adversary gains additional opportunities to interfere. We give an example showing that cloning can indeed introduce insecurities into a previously secure protocol.

\medskip

\hspace*{1cm}\begin{minipage}{4cm}
\begin{tikzpicture}
 \node at   (0,0) (g) {$g$};
 \node at (-1,-0.2) (dummy) {$f_3$};
 \node at (-1.5  , -1.5)  (f1)    {$f_1$};
 \node at (1.5,    -1.5)  (f2)    {$f_2$};
 
 \node at (-1.8,-2.5)   (x)    {$x$};
 \node at (-1,-2.5)   (y)    {$y$};
 \node at (1.1,-2.5)    (u)    {$u$};
 \node at (1.9,-2.5)    (v)    {$v$};
 
 \draw [->] (x) edge [bend left=20]  (f1);
 \draw [->] (y) edge [bend right=20]  (f1);
 \draw [->] (u) edge [bend left=20]  (f2);
 \draw [->] (v) edge [bend right=20]  (f2);
 \draw [->] (f1) edge [bend left] (dummy);
 \draw [->] (dummy) edge [bend left] (g);
 \draw [->] (f1) edge [bend right] (g);
 \draw [->] (f2) edge [bend right] (g);
\end{tikzpicture}
\captionof{figure}{Protocol $\protocol_{\mathtext{sec}}$}
\label{figure:secure protocol}
\end{minipage}\ \ \ \ 
\begin{minipage}{4cm}
 \begin{tikzpicture}
 \node at   (0,0) (g) {$g$};
 \node at (-1,-0.2) (dummy) {$f_3$};
 \node at (-1.5  , -1.5)  (f1)    {$f_1$};
 \node at (-0.5  , -1.5)  (f1clone)    {$f'_1$};
 
 \node at (1.5,    -1.5)  (f2)    {$f_2$};
 
 \node at (-1.3,-2.5)   (x)    {$x$};
 \node at (-0.5,-2.5)   (y)    {$y$};
 \node at (1.1,-2.5)    (u)    {$u$};
 \node at (1.9,-2.5)    (v)    {$v$};

 \draw [->] (x) edge [bend left=20]  (f1);
 \draw [->] (y) edge [bend right=20]  (f1);
 \draw [->] (x) edge [bend left=20]  (f1clone);
 \draw [->] (y) edge [bend right=20]  (f1clone);
 
 \draw [->] (u) edge [bend left=20]  (f2);
 \draw [->] (v) edge [bend right=20]  (f2);

 \draw [->] (f1) edge [bend left] (dummy);
 \draw [->] (dummy) edge [bend left] (g);
 \draw [->] (f1clone) edge [bend left] (g);
 \draw [->] (f2) edge [bend right] (g);
\end{tikzpicture}
\captionof{figure}{Protocol $\protocol_{\mathtext{insec}}$}
\label{figure:insecure protocol}
\end{minipage}

\medskip

Consider the protocols $\protocol_{\mathtext{sec}}$ and $\protocol_{\mathtext{insec}}$ from Figures~\ref{figure:secure protocol} and~\ref{figure:insecure protocol}. 
$\protocol_{\mathtext{insec}}$ is the result of cloning $f_1$ in $\protocol_{\mathtext{sec}}$, the copy is called $f_1'$. In the original protocol $\protocol_{\mathtext{sec}}$, the answer to the query of $f_1$ is used as input for both the $f_3$ and the $g$-queries. In contrast, the protocol $\protocol_{\mathtext{insec}}$ contains the same query to $f_1$ (whose result is used in the later $f_3$-query) and a further query to $f_1'$, using the same input values as the $f_1$-query (namely $x$ and $y$). The result of this query is  used as input for the $g$-query. This change to the protocol introduces insecurity: While $\protocol_{\mathtext{sec}}$ is secure, the protocol $\protocol_{\mathtext{insec}}$ is insecure. Hence cloning preserves insecurity, but does not preserve security.

To see that $\protocol_{\mathtext{sec}}$ is secure, first note that the node $f_3$ can be removed without affecting security. In the resulting protocol, $g$ receives two copies of the result of $f_1(x,y)$, which is equivalent to $g$ only getting a single copy of that bit. This protocol is identical to the one from Figure~\ref{figure:simple nosync example}, and hence $\protocol_{\mathtext{sec}}$ is secure by Example~\ref{example:simple secure protocol}.

On the other hand, $\protocol_{\mathtext{insec}}$ has a tracking strategy given by $T=\set{t_{init},g}$, with $t_{init}=f_2$. Therefore, $\protocol_{\mathtext{insec}}$ is insecure due to Theorem~\ref{theorem:tracking strategies work}. 

This shows that introducing an additional query into a secure protocol can render the protocol insecure, even if the exact same set of input values for the new query is already used for an existing query in the protocol. Therefore, queries to an adversial network should be kept to a minimum.

\subsection{Security Proofs for Deeper Protocols}\label{sect:embedding application}

We now use Theorem~\ref{theorem:embedding}, to ``lift'' the results obtained for flat protocols in Section~\ref{sect:two cases flat secure protocols} to protocols of arbitrary depth.
For simplicity, we only consider \emph{layered} protocols $\protocol$, i.e., $\protocol=L_0\cup\dots\cup L_n$, where $L_i\cap L_j=\emptyset$ for $i\neq j$, each predecessor of a node in $L_i$ is in $L_{i+1}$, and variable nodes only appear in $L_n$. One can easily rewrite every protocol into a layered one without affecting security. Our first result in this section generalizes our security result for flat protocols that consist of two ``disjoint'' components (Theorem~\ref{theorem:disjoint variables}) to protocols of arbitrary depth:

\begin{restatable}{corollary}{corollarydisjointvariablesgeneralization}\label{corollary:disjoint variables generalization}
 Let $\protocol$ be a layered protocol with levels $L_0,\dots,L_n$ with some $i$ such that $L_i=I_1\cup I_2$ with $\var{I_1}\cap\var{I_2}=\emptyset$ and $\card{\var{I_1}}>\card{I_1}$, $\card{\var{I_2}}>\card{I_2}$. Then $\protocol$ is secure.
\end{restatable}

Our result on flat protocols $g(f_1(\overrightarrow{x_1}),\dots,f_n(\overrightarrow{x_n}))$ where each of the $f_i$ has a private variable can be generalized as follows. Note that for such a protocol, a tracking strategy exists if and only if there is one service $f_i$ that has access to all variables except the one private variable for each other $f_j$ ($f_i$ itself may have more than one private variable), i.e., if there is one node $f_i$ that ``sees'' all variables except for $n-1$ many private ones.

\begin{restatable}{corollary}{corollaryprivatevariablesgeneralization}\label{corollary:private variables generalization}
 Let $\protocol$ be a layered protocol with levels $L_0,\dots,L_n$ with some $i$ such that for each $f\in L_i$, there is a variable $x_f\in\vars(f)\setminus\var{L_i\setminus\set{f}}$ and there is no $f\in L_i$ with $\var{f}\supseteq\var{L_i}\setminus\set{x_{f'}\ \vert\ f'\in L_i}$. Then $\protocol$ is secure.
\end{restatable}

Both of these results follow easily by embedding a flat protocol satisfying the prerequisites of Theorem~\ref{theorem:disjoint variables} or~\ref{theorem:each function has private variable} into $\protocol$.

\section{Generalized Security}\label{sect:generalized security}

In many situations, our security definition is too weak. In an internet shopping protocol, we want the adversary to be unable to link the user's address and the product, but probably do not care about whether she can link these values to the bit indicating express delivery. Yet, our  definition deems a protocol secure as soon as the latter bit cannot be linked to the former values. Similarly, one might be tempted to apply Theorem~\ref{theorem:embedding} to make a protocol $\protocol$ secure by embedding into it a known secure protocol $\protocol'$, applied to dummy variables. Clearly, this does not give any meaningful security for $\protocol$, as linking the dummy variables is irrelevant.

This illustrates that often, we consider a protocol $\protocol$ insecure as soon as the adversary can link a subset of ``relevant'' variables $X_r\subseteq\var{\protocol}$. We call such a  protocol $X_r$-insecure, and $X_r$-secure otherwise. This condition is clearly stronger than security in the sense of  Definition~\ref{definition:secure protocol}, since every $X_r$-secure protocol trivially is secure. Our results can be directly applied to $X_r$-security via the following straightforward relationship:

\begin{proposition}\label{prop:partial insecurity characterization}
 Let $\protocol$ be a protocol and let $X_r\subseteq\var{\protocol}$. Then $\protocol$ is $X_r$-secure if and only if $\restr{\protocol}{X_r}$ is secure, where $\restr{\protocol}{X_r}$ is obtained from $\protocol$ by removing all variables $x\in\var{\protocol}\setminus X_r$ and their outgoing edges.
\end{proposition}

Similarly, there are natural generalizations of our security defintion such as not simply tracking the first session in a global session, but waiting for a session with ``interesting'' values for user variables (values of these variables must be present at the node $t_{init}$ in a tracking strategy), 
and clearly an arbitrary number of sessions can be tracked consecutively by ``recycling'' the session cookie after the successful linking of a session. Our techniques can easily be applied to these situations as well.
 
\section{Conclusion}\label{sect:conclusion}

We have initiated the study of active linking attacks on communication protocols which exploit that the adversary has control over user data. We introduced a model, formalized a security definition and gave a sound criterion---tracking strategies---for insecurity of protocols. We also gave sound criteria (via embedding) for security of protocols. For a large class of protocols, these criteria are complete, i.e., they detect all insecure protocols. The question whether this completeness holds in general remains open.

Further interesting questions for future research concern relaxing the  ``worst-case'' assumptions underlying our approach:

We allow the adversary to ignore the intended semantics of the webservices under her control, and to choose arbitrary query replies for a linking attack. Realistically, the adversary will try to honestly answer a majority of the queries in order to avoid detection. Further, we require the adversary to be successful for every possible schedule. Realistically, while interleaving certainly happens (and can be enforced to some extent by a privacy-aware user), an adversary strategy that is only successful against schedules with a limited amount of interleaving might be enough to consider a protocol insecure. 

Other extensions include randomized adversary strategies, which only have to be successful with a certain probability, and scenarios where where only a subset of the web services is controlled by the adversary.

\begin{appendix}
   
\section{Proof of Theorem~\ref{theorem:tracking strategies work}}

\theoremtrackingstrategieswork*

\begin{proof}
 \resetprooffacts
We sketch how to construct a strategy $\Pi$ which does tracking for nodes in $T$ and uses the other nodes to transmit bits through the protocol unmodified. Without loss of generality we assume that $T$ is $\leadsto$-closed, i.e. it contains all nodes reachable from $t_{init}$.
For every variable node $x$ we fix a path $p_x$ as mentioned in the cover condition of definition \ref{definition:generalized tracking strategy}.
The adversary can construct a function $p : \tau \setminus \left(\vars(\tau) \cup \set{t_{init}}\right) \to \tau \cup \set{\bot}$ with the following properties: For $u \in T$ let $p(u)=v$ with $v \rightarrow u$ and $v \in T$, for $u \not\in T$ let $p(u)$ be the predecessor of $u$ if for some variable node $x$ there is a path $p_x$ which ends in $u$, and let $p(u)=\bot$ otherwise. This function expresses whether a node will echo back one of its inputs as a reply, and if so, which one. Note that $(T \setminus \set{t_{init}}) \subseteq \mathit{dom}(p)$; this is no contradiction, since for a node $u \in T$, $u \neq t_{init}$ projection to the input parameter supplying the earlier generated tracking cookie is indistinguishable from generating an ``own'' tracking cookie - there's effectively only one cookie generated by $t_{init}$ and then passed on.

Now let $\beta$ be an adversary view. First we assume the first local session to address $t_{init}$ still has pending unanswered queries. Let $(u,I,\bot,\bot)$ be the earliest unanswered query in $\beta$. If $u=t_{init}$ we set $\Pi(\beta)=\mathsf{reply}(k,r)$, where $k$ is the index of that query and $r=0$ if there's no query to $t_{init}$ in the prefix of $\beta$ up to $(u,I,\bot,\bot)$ and $r=1$ otherwise. This models the tracking behavior of $t_{init}$ with inconsistent replies. In the case of $u\neq t_{init}$, we make a further distinction: If $p(u)=\bot$ the node $u$ isn't vital to the strategy and we just set $\Pi(\beta)=\mathsf{reply}(k,1)$; if we have $p(u)\neq \bot$ however, we set $\Pi(\beta)=\mathsf{reply}(k, I(p(u)))$, each with an appropriate choice of $k$ as above. Note that the symbol $I$ in the latter case is the partial assignment from the view.
Now assume the first local session to ask $t_{init}$ has been served completely. In this case, we set $\Pi(\beta)=\mathsf{print}(I_{out})$ with $I_{out}$ extracted from $\beta$ as follows: Let $(u,I,r,t)$ an element of $\beta$ with $u \in T$, $r=0$ and there is a variable node $x$ such that $p_x$ ends in $u$. Then let $I_{out}(x)=I(x)$.

It remains to show that $\Pi$ is successful. Since the session to ask $t_{init}$ first will at some point be complete, a view will finally be attained which makes $\Pi$ emit a $\mathsf{print}$ action. $\Pi$ will then produce the output $I_{out}$. We now show that for a given global session $(S,\sigma)$ we have $I_{out} \in S$. Assume $S=\set{I_1, \ldots, I_N}$ and let $I_{t_0}$ be the session which asks $t_{init}$ first, i.e. the session the adversary wants to track.
\begin{prooffact}
\label{prooffact:tracking works 1}
If $u \in T$ then $I_i(u)=0$ if and only if $i=t_0$.
\end{prooffact}
\begin{proof}
Assume there is a query which violates this property. Let $(i,t)$ with $t \in T$ be the earliest such query. Then clearly $t \neq t_{init}$ by the construction of $\Pi$. We have $p(t)=v$, where $v \in T$ and $v \rightarrow t$ by definition of $p$. Therefore the node $v$ has already been asked and since $(i,t)$ is the earliest counterexample, we get $I_i(v)=0$ if and only if $i=t_0$, i.e. the statement holds for $v$. Because of $v \rightarrow t$ we also get $I_i(v)=I_i(t)$, i.e. $v$'s reply is an input to $t$, and this yields $I_i(t)=0$ if and only if $i=t_0$, contradicting our assumption. Therefore no query can violate the stated property.
\end{proof}
With the above Fact we can justify calling $t_{init}$ a synchronizer. It states that, given a tracking strategy $T$, all nodes in $T$ really know the unique session to be tracked. With the next Fact we show that the adversary can correctly channel variable contents to a tracking node.
\begin{prooffact}
\label{prooffact:tracking works 2}
For a variable node $x$, let $p_x$ be of the form $x \leadsto u \leadsto t$, where $x \in \vars(\tau)$, $u \not\in T$ and $t \in T$. Then $I_i(x)=I_i(u)$ when local session $I_i$ has asked node $u$.
\end{prooffact}
\begin{proof}
Assume $I_i(x) \neq I_i(u)$ and let $v$ be the first node in the path with $I_i(x) \neq I_i(v)$. By definition of function $p$ we get $p(v)=w$ for a node $w$ with $w \rightarrow v$. Since $w$ is a predecessor of $v$, we have $I_i(x)=I_i(w)$. Now consider a view in which $v$ is the next node to be asked by $I_i$; according to the definition of $\Pi$, the adversary will emit $\mathsf{reply}(k, I_i(p(v))) = \mathsf{reply}(k, I_i(w)) = \mathsf{reply}(k, I_i(x))$ for some $k$. We must have $I_i(x)=I_i(v)$ prior to asking $u$, which contradicts our assumption. Therefore no such $v$ can exist.
\end{proof}
We combine both Facts to obtain the result that tracking nodes get to know the input values of session $I_{t_0}$.
\begin{prooffact}
\label{prooffact:tracking works 3}
For a variable node $x$,  let $p_x$ be of the form $x \leadsto u \rightarrow t$, where $x \in \vars(\tau)$, $u \not\in T$, $t \in T$ and $t$ replies with a zero when asked. Then $I_i(u)=I_{t_0}(x)$, i.e. $u$ will supply $t$ with a value from the session to be tracked.
\end{prooffact}
\begin{proof}
Since $t \in T$ and $I_i(t)=0$ after $t$ has been asked, we can apply Fact \ref{prooffact:tracking works 1} which shows that $i=t_0$ must hold. Using Fact \ref{prooffact:tracking works 2} we can see that $I_{t_0}(x)=I_{t_0}(u)=I_i(u)$.
\end{proof}
We can now prove that $\Pi$ is indeed a successful strategy. Let $\beta$ be a view in which session $I_{t_0}$ has been completed. Then $\Pi(\beta)=\mathsf{print}(I_{out})$ and by construction of $\Pi$ we see that $I_{out}$ is the union of all $I$ such that $(u,I,r,t)$ is an element in this view with $u \in T$, $r=0$ and there is a variable node $x$ such that $p_x$ ends in $u$. Using Fact \ref{prooffact:tracking works 3} we see that $I_{out} \subseteq I_{t_0}$. Since we have a path $p_x$ for every variable node $x$, we know that $I_{out}$ contains an assignment for every input value.
Therefore $I_{out}=I_{t_0} \in S$.

\end{proof}

\section{Finding Tracking Strategies in Polynomial Time}

We now prove that determining whether there exists a tracking strategy for a given protocol $\protocol$ can be done in polynomial time.

\decisionproblemwidth{10cm}{\exstrat}{Protocol $\protocol$}{Is there a tracking strategy for $\protocol$?}

\begin{theorem}
 The problem \exstrat\ can be solved in polynomial time.
\end{theorem}

\begin{proof}
  Clearly, a tracking strategy exists if and only if there is some note $t_{init}\in\protocol$ such that for the set $T_{t_{init}}=\set{u\in\protocol\ \vert\ t_{init}\rightsquigarrow u}$, the cover condition is satisfied. We test this property for each non-variable node $t_{init}$ with a standard application of network flow. For this, we modify $\protocol$ to obtain a network-flow instance as follows:
  
  \begin{itemize}
   \item Contract all nodes $T$ into a single node $t$, i.e., introduce a new node $t$, then for each edge $(u,v)$ with $u\notin T$ and $v\in T$, add an edge $(u,t)$, and finally remove all nodes in $T$.
   \item Introduce a node $X$ and edges $X\rightarrow x$ with capacity $1$ for each $x\in\var{\protocol}$.
   \item For each node $u\in\protocol$, introduce two new nodes $u_{in}$ and $u_{out}$. Replace all incoming edges $(v,u)$ with an edge $(v,u_{in})$ (with capacity $1$), and all outgoing edges $(u,v)$ with an edge $(u_{out},v)$ (also with capacity $1$). Add an edge $(u_{in},u_{out})$ with capacity $1$.
  \end{itemize}
  
  Then the cover property is satisfied if and only if there is a network flow from $X$ to $T$ with value $\card{\var{\protocol}}$:
  
  \begin{itemize}
   \item Assume that there is a network flow of value $\card{\var{\protocol}}$. Then the flow uses each outgoing edge of $X$, since there are $\card{\var{\protocol}}$ many of these, and each has capacity $1$. Since each internal edge $(u_{in},u_{out})$ has capacity $1$, it can be used only once in the network flow. Therefore, the flow must consist of $\card{\var{\protocol}}$ node-disjoint paths to $t$, which can be translated into node-disjoint paths into the set $T$. Hence the cover condition is satisfied.
   \item For the converse, assume that the cover condition is satisfied. Then there are node-disjoint paths from each $x\in\var{\protocol}$ to a node in $T$, which correspond to node-disjoint paths from each $x$ to $t$ in the modified protocol. By extending these paths with edges $(X,x)$, they clearly form a network flow from $X$ to $T$ of value $\card{\var{\protocol}}$.
  \end{itemize}

\end{proof}

\section{Security Proofs for Flat Protocols}

In this section we present the proof of the theorems from Section~\ref{sect:two cases flat secure protocols}. First, we give a simple criterion for the existence of tracking strategies for flat protocols.

\begin{proposition}\label{prop:flat tracking strategy characterization}
 Let $\protocol=g(f_1(\overrightarrow{x_1}),\dots,f_n(\overrightarrow{x_n}))$ be a flat protocol. Then a tracking strategy for $\protocol$ exists if and only if there is a function $t\colon\set{1,\dots,n}\rightarrow\vars{\protocol}$ such that 
 \begin{itemize}
  \item $\cup_{i=1}^n t(i)=\vars{\protocol}$,
  \item there is at most one $i$ with $\card{t(i)}\ge 2$.
 \end{itemize}
 In particular, if such a function exists, then $\protocol$ is insecure.
\end{proposition}

For a flat protocol, we also call a function $t$ as above a \emph{flat tracking strategy}. For flat protocols, flat tracking strategies and tracking strategies directly correspond to each other with $t_{\text{init}}=f_i$, where $i$ is the unique number with $\card{t(i)}\ge2$ if such an $i$ exists, and $i=1$ otherwise. This, together with an application of Theorem~\ref{theorem:tracking strategies work} proves Proposition~\ref{prop:flat tracking strategy characterization}.

\subsection{Proof of Theorem~\ref{theorem:disjoint variables}}\label{sect:secure flat:disjoint}

\theoremdisjointvariables*

\begin{proof}
 \resetprooffacts
  The main idea of the proof is as follows: To show that every adversary-strategy $\Pi$ fails, we construct, depending on $\Pi$, a global session $(S,\sigma)$ and then for each local session $s\in S$, a global session $(S_s,\sigma_s)$ such that the adversary-view for both sessions is identical, but $s\notin S_s$. If $\Pi$ fails on $(S,\sigma)$, then $\Pi$ is not a winning strategy. If $\Pi$ succeeds on $(S,\sigma)$, then $\Pi$ prints a correct local session $s\in S$. Since the global session $(S_s,\sigma_s)$ is indistinguishable from $(S,\sigma)$ for the adversary, the strategy $\Pi$ prints the same session $s$ for the global session $(S_s,\sigma_s)$. Since $s\notin S_s$, the strategy $\Pi$ thus fails on $(S_s,\sigma_s)$, and we have shown that $\Pi$ is not a winning strategy.
  
 Hence let $\Pi$ be an adversary strategy. Without loss of generality, we assume that $\card{\var{I_1}}\ge\card{\var{I_2}}$, and thus $\card{\var{I_1}}=\card{\var{I_2}}+d$ for some $d\ge0$. We fix some terminology for our construction:
 \begin{itemize}
  \item For $j\in\set{1,2}$, an \emph{$I_j$-assignment} is a function $a\colon\var{I_j}\rightarrow\set{0,1}$. Note that a local session for the protocol $\protocol$ is a union of an $I_1$-assignment and an $I_2$-assignment, and, since $\var{I_1}\cap\var{I_2}=\emptyset$, any such union is a local session.
  \item For a local session $s$ and $j\in\set{1,2}$, the \emph{$I_j$-component} of $s$ is the restriction of $s$ to $\var{I_j}$. 
 \end{itemize}
 
 In the proof, we only consider global sessions with the following properties:
  
 \begin{itemize}
  \item there are exactly $2^{\card{\var{I_1}}}$ local sessions,
  \item for each $I_1$-assignment $a$, there is exactly one local session $s$ such that the $I_1$-component of $s$ is $a$,
  \item for each $I_2$-assignment $a$, there are exactly $2^d$ local sessions $s$ such that the $I_2$-component of $s$ is $a$,
  \item the schedule first performs all queries to $f_1$, then all questions to $f_2$, etc, the $g$-questions are scheduled last. For each service, the questions are scheduled in lexicographical order of the arguments.
 \end{itemize}

 In particular, for each $i\in\set{1,\dots,n}$, the $f_i$-queries are \emph{all} possible assignments to the variables in $\overrightarrow{x_i}$, asked in lexographical order, each assignment possibly several times. Let $A_1$ be the set of $I_1$-assignments, and let $A_2$ be the set containing $2^d$ distinct copies of each $I_2$-assignment (hence $\card{A_1}=\card{A_2}=2^{\card{\var{I_1}}}$). All global sessions satisfying the above generate the same adversary-view until the first $g$-query is performed. In particular, the strategy $\Pi$ will return the same answers to the $f_i$-questions for all global sessions with the above structure. Since adding edges to an insecure protocol leaves the protocol insecure, we can assume that for $i\in I_j$, the input variables to $f_i$ in $\protocol$ is the entire set $\var{I_j}$. Hence we use the following notation: For an element $a\in A_j$ and $i\in I_j$, with $f_i(a)$, we denote the values returned by the adversary for the $f_i$-question when asked for the assignment $a$, and stress that these values are identical for all global sessions following the above pattern. 
 
 For $j\in\set{1,2}$, we say that elements $a_1,a_2\in A_j$ are $j$-equivalent and write $a_1\sim_j a_2$ if $f_i(a_1)=f_i(a_2)$ for all $i\in I_j$.
 We note that the adversary can give different replies to different copies of the same $I_2$-assignment, i.e., even if $a_1$ and $a_2$ are copies of the the same $I_2$-assignment, then $f_i(a_1)\neq f_i(a_2)$ may still hold for $i\in I_2$, and hence for these elements we may have $a_1\not\sim_2 a_2$. This issue does not occur for elements of $A_1$, since $A_1$ contains only one copy of each $I_1$-assignment.
 
 Since each $\sim_j$-equivalence class is determined by the values of the functions $f_i$ for $i\in I_j$, there are at most $2^{\card{I_j}}$ $\sim_j$-equivalence classes. Our construction relies on the existence of a sufficient number of elements $a\in A_j$, for both $j\in\set{1,2}$, such that there is an element $a'\in A_j$ with $a\sim_ja'$, and $a'$ is a copy of a different $I_j$-assignment than $a$. We call such an element \emph{$j$-ambigious}, since the $I_j$-assignment represented by $a$ is not uniquely determined by its $\sim_j$-equivalence class. These elements are helpful for the following reason: When the adversary receives a question for the service $g$ from a given local session, this question contains the answers given by the adversary for all $f_i$-question of the same local session, and hence the adversary ``sees'' the $\sim_j$-equivalence class of both $I_j$-components of this local session. If one of these components is $j$-ambigious, the adversary cannot determine the $I_j$-component directly from the inputs to $g$, and hence cannot directly reconstruct the local session from the inputs received for the $g$-question.
  
 We now show that there are ``enough'' ambigious elements for our purposes:
  
 \begin{prooffact}\label{prooffact:number of ambigious elements}
  For each $j\in\set{1,2}$, more than half of the elements in $A_j$ are $j$-ambigious.
 \end{prooffact}
 
 \begin{proof}
   This follows with a simple counting argument. To treat the cases $j=0$ and $j=1$ uniformly, let $d_0=0$ and $d_1=d$. Since the number of $\sim_j$-equivalence classes is $2^{\card{I_j}}$, and each $I_j$-assignment occurs $2^{d_j}$ times in $A_j$, there can be at most $2^{\card{I_j}}\cdot 2^{d_j}=2^{\card{I_j}+d_j}$ non-$j$-ambigious elements. In particular, some elements are $j$-ambigious, and thus there is at least one $\sim_j$-equivalence class containing elements that represent different $I_j$-assignments. Therefore, there can be at most $2^{\card{I_j}}-1$ equivalence classes containing non-$j$-ambigious elements, and therefore there are at most $(2^{\card{I_j}}-1)\cdot 2^{d_j}\leq2^{\card{I_j}+d_j}-1$ non-$j$-ambigious elements. Since $\card{A_j}=2^{\card{\var{I_j}}+d_j}$ and $\card{\var{I_j}}=\card{I_j}+e$ for some $e>0$, it follows that $\card{A_j}=2^{\card{I_j}+d_j+e}$, hence less than half of the elements of $A_j$ are non-$j$-ambigious as claimed.
 \end{proof}
 
  We now construct the global session $(S,\sigma)$. Since $\var{I_1}\cap\var{I_2}=\emptyset$, a global session matching the above criteria is uniquely determined by a bijection between $A_1$ and $A_2$, and two such global sessions result in the same view for the adversary if and only if the same questions are asked for the service $g$. We represent the bijection as a table with $2^{\card{\var{I_1}}}$ rows and two columns, where the first column contains elements from $A_1$, and the second column contains elements from $A_2$. The bijection determining the global session then relates the elements that occur in the same row. The set $S$ of local sessions can simply be read off the rows of the table in the natural way. We therefore identify local sessions $s\in S$ and rows in the table. The $g$-queries resulting from the session represented in the table can also be read off the table: For a local session $s$ consisting of the elements $a_1\in A_1$ and $a_2\in A_2$, the input to the $g$-query from that session is the tuple $(f_1(a_1),\dots,f_k(a_1),f_{k+1}(a_2),\dots,f_n(a_2))$, if $I_1=\set{1,\dots,k}$ and $I_2=\set{k+1,\dots,n}$.
  
  For the global session $(S,\sigma)$, the elements of $A_1$ and $A_2$ are now distributed in the table as follows:
 
 \begin{itemize}
  \item In the first column, we first list all elements of $A_1$ that are $1$-ambigious, ordered by equivalence class, and then all elements that are not $1$-ambigious. Due to Fact~\ref{prooffact:number of ambigious elements}, more than half of the rows contain $1$-ambigious elements in their first column.
  \item We distribute the $A_2$-elements in the second column such that the rows in the lower half and the last row of the upper half of the table contain only $2$-ambigious elements in their second column. This is possible since due to Fact~\ref{prooffact:number of ambigious elements}, a majority of the elements in $A_2$ are $2$-ambigious. We re-order the elements in the upper half of the table such that no row in the upper half has a neighboring row with an element representing the same $I_2$-assignment. This is possible since $\card{\var{I_2}}>\card{I_2}\ge1$, and thus there are at least $4$ different $I_2$-assignments, and hence no $I_2$-assignment can occur in more than half of the rows in the upper half of the table.
 \end{itemize}

  By construction, the ambigious elements make up at least half of each column. Hence each row in the table contains an $1$-ambigious element in the first column (this is true at least for the upper half of the rows), or a $2$-ambigious element in the second column (this is true at least for the lower half of the table). 
 
 We now construct, for each $s\in S$, the session $(S_s,\sigma_s)$. For this it suffices to define the set $S_s$, the schedule $\sigma_s$ is then determined by the above criteria. The set $S_s$ is obtained from the above table-representation of $S$ as follows:
 
 \begin{itemize}
   \item If the row $s$ has a $2$-ambigious $A_2$-element $a$ in the second column, we swap the positions of $a$ and an element $a'$ that represents a different $I_2$-assignment than $a$, and for which  $a'\sim_{\Pi,2}a$ holds. Such an element exists since $a$ is $2$-ambigious. Clearly, the $I_1$-component of $s$ is paired with a different $I_2$-component after the swap, and hence the local session $s$ is removed from $S$.
   \item Otherwise, $s$ appears in the upper half of the table, and is not the last row in the upper half of the table. In particular, $s$ has a $1$-ambigious $A_1$-element $a$ in the first column. Since the $A_1$-elements appear ordered by equivalence class, one of the rows directly above or below $s$ contains an element $a'\neq a$ with $a'\sim_1a$ in the first column. Since $s$ is not the last row of the upper half of the table, $a'$ also appears in the upper half of the table. We now swap $a$ and $a'$ in the first column. By distribution of the $A_2$-elements in the upper half of the second column, the two rows involved in the exchange contain elements corresponding to different $I_2$-assignments, and since each $I_1$-assignment appears only once in $A_1$, we also know that $a$ and $a'$ represent different $I_1$-assignments. Therefore, as above, the $I_1$-component of $s$ is paired with a different $I_2$-component after the swap, and hence again, the local session $s$ is removed from $S$.
 \end{itemize}

 In both cases, we removed one occurrence of the local session $s$ from $S$. Since each $I_1$-component appears only once in $S$, the local session $s$ also appears only once in $s$, and hence we removed the only appearance of $s$ from $S$. Therefore, the new global session $(S_s,\sigma_s)$ does not contain any local session $s$ anymore. Since the session $(S_s,\sigma_s)$ is obtained from $(S,\sigma)$ by exchanging some element from $A_j$ with an $\sim_j$-equivalent element, each row in the table generates the same questions to the $g$-service as the session $(S,\sigma)$, and therefore, the set of $g$-queries in the global session $(S_s,\sigma_s)$ is the same as in the session $(S,\sigma)$. Since the individual queries are performed in lexicographical order, the adversary-view of $(S,\sigma)$ and $(S_s,\sigma_s)$ is identical, which concludes the proof.
\end{proof}

\subsection{Proof of Theorem~\ref{theorem:each function has private variable}}\label{sect:secure flat:private variables}

In this section, we show that a flat protocol $\protocol=g(f_1(x^1_1,\dots,x^1_{m_1}),f_2(x^2_1,\dots,x^2_{m_2}),\linebreak\dots,f_n(x^n_1,\dots,x^n_{m_n}))$ where every $f_i$ has at least one private variable is insecure if and only if there is a tracking strategy for $\protocol$. We know that tracking strategies always imply insecurity of the protocol, hence the interesting part of the result is to show that in the absence of a tracking strategy, the protocol is secure. For our proof, we introduce some notation: For each $f_i$, let $X_i$ be the set of private variables for $f_i$. Clearly, $X_i\cap X_j=\emptyset$ if $i\neq j$. Let $V$ be the set of non-private variables, i.e., $V=\var{\protocol}\setminus\cup_{i=1}^n X_i$.
 
We first cover a number of ``simple'' cases:

\begin{lemma}\label{lemma:simple cases}
  Let $\protocol$ be a flat protocol for which no tracking strategy exists and which satisfies one of the following conditions:
  \begin{enumerate}
   \item\label{lemma:simple cases:all X1 singletons} There is some $i$ with  $\card{X_i}\ge2$,
   \item \label{lemma:simple cases:not all variables from v seen by the same fi} there is an $i\in\set{1,\dots,n}$ with $V\subseteq\var{f_i}$,
   \item $\card{V}\leq 2$.
  \end{enumerate}
  Then $\protocol$ is secure.
\end{lemma}

\begin{proof}
  Without loss of generality assume that $\card{X_1}\ge\card{X_2}\ge\dots\ge\card{X_n}\ge1$. We first consider the case $\card{X_2}\ge 2$. In this case, we also have $\card{X_1}\ge2$. Clearly, no tracking strategy can exist. To show that the protocol is secure, assume indirectly that it is insecure. We drop all variables from $V$, the resulting protocol $\protocol'$ is still insecure. (This can easily be seen as a consequence of Theorem~\ref{theorem:embedding}, or simply by observing that the adversary can simulate all ``missing'' variables by using $0$-values). However, we can choose $I_1=\set{1}$ and $I_2=\set{2,\dots,n}$, then security of the protocol follows from Theorem~\ref{theorem:disjoint variables}, since (using the notation of the theorem), $\var{I_1}=X_1$ and $\var{I_2}=X_2\cup\dots\cup X_n$, with $\card{\var{I_1}}\ge2$, and $\card{\var{I_2}}\ge n$. Therefore, for the remainder of the proof we assume that $\card{X_2}=\dots=\card{X_n}=1$.
  
  Next, we consider the case $V\subseteq\var{f_1}$. In this case, a flat tracking strategy exists with $t(1)=X_1\cup V$ and $t(i')=X_{i'}$ for all $i'\neq i$. The theorem thus follows with Proposition~\ref{prop:flat tracking strategy characterization}. From now on, we assume that $V\not\subseteq\var{f_1}$, in particular, this implies $V\neq\emptyset$.
  
  We now prove the lemma:

  \begin{enumerate}
    \item In this case, it follows that $\card{X_1}\ge2$. Due to the above, we can assume that there is a variable $v\in V\setminus\var{f_1}$. We use Theorem~\ref{theorem:disjoint variables} to prove security of $\protocol$, with $I_1=\set{1}$ and $I_2=\set{2,\dots,n}$. We have that $\var{I_1}\supseteq X_1$ and thus $\card{\var{I_1}}>\card{I_1}=1$, and $\var{I_2}\supseteq\set{v}\cup \cup_{i=2}^n X_i$, hence $\card{\var{I_2}}\ge n>n-1=\card{I_2}$.
    \item Due to point~\ref{lemma:simple cases:all X1 singletons}, we can assume that $\card{X_i}=1$ for $1\leq i\leq n$. If there is an $i\in\set{1,\dots,n}$ with $V\subseteq\var{f_i}$, then there is a tracking strategy $t$ for $\protocol$, defined as $t(i)= X_i\cup V$ and $t(i')=X_{i'}$ for $i'\neq i$, but from the prerequisites we know that there is no tracking strategy for $\protocol$.
    \item Due to the above, we can assume that $V\ge 2$, since otherwise there will be some $i$ with $V\subseteq\var{f_i}$. Thus assume that $\card{V}=2$, and let $V=\set{v_1,v_2}$. Again due to point~\ref{lemma:simple cases:not all variables from v seen by the same fi}, we know that there is no $i$ with $v_1,v_2\in\var{f_i}$. We show that the protocol is secure using Theorem~\ref{theorem:disjoint variables}. To apply the theorem, we define $I_1=\set{i\in\set{1,\dots,n}\ \vert\ v_1\in\var{f_i}}$, and $I_2=\set{1,\dots,n}\setminus I_1$. Note that $v_2\in\var{I_2}$, and hence $\var{I_j}=\cup_{i\in I_j}X_i\cup\set{v_j}$ for both $j\in\set{1,2}$. In particular, $\card{\var{I_j}}>\card{I_j}$ holds for both $j$. Security of the protocol thus follows from Theorem~\ref{theorem:disjoint variables}.
  \end{enumerate}
\end{proof}

By Lemma~\ref{lemma:simple cases}, we can in particular assume that for each $f_i$, there is a single private variable. For protocols of this form, there is a trivial characterization of the cases in which a tracking strategy exists:

\begin{lemma}\label{lemma:tracking strategies for private variables}
 If $\protocol$ is a protocol where each $f_i$ has exactly one private variable, then there is a tracking strategy for $\protocol$ if and only if there is some $i$ with $V\subseteq\var{f_i}$.
\end{lemma}

\begin{proof}
 This easily follows from Proposition~\ref{prop:flat tracking strategy characterization}: If an $i$ satisfying the condition exists, a flat tracking strategy $t$ is given by $t(i)=\var{f_i}$, and $t(i')=X_{i'}$ for all $i'\neq i$. On the other hand, assume that a flat tracking strategy $t$ exists. Since each $X_i$ contains exactly one private variable for $f_i$, we know that $X_i\subseteq t(i)$ for all $i$. Since there is at most one $i$ with $\card{t(i)}>1$, for this $i$ we must have $V\subseteq\var{f_i}$. If there is no such $i$, then $V=\emptyset\subseteq\var{f_i}$.
\end{proof}

In the security proof, it will be convenient to make some assumptions about the structure of the protocol $\protocol$. We call protocols meeting these assumptions a protocol in \emph{normal form}. The main part of the proof is showing that protocols in normal form for which no tracking strategy exists are secure, from this we will later easily deduce that the result also holds for all flat protocols where each service-node has a private variable.

\begin{definition}
 A flat protocol $\protocol$ where each $f_i$ has a single private variable is in \emph{normal form}, if the following conditions are satisfied:
 \begin{itemize}
  \item For each $i\in\set{1,\dots,n}$, we have that $\var{f_i}=X_i\cup \left(V\setminus{v_i}\right)$ for some $v_i\in V$,
  \item for each $v\in V$, there is some $i$ with $v\notin\var{f_i}$,
  \item $\card{V}\ge3$.
 \end{itemize}
 In particular, in this case we have $\card{\var{f_i}}=\card{V}$ for all $i$ and $n\ge\card{V}$.
\end{definition}

The main ingredient of the proof of our classification is the result that the classification holds for protocols in normal form:

\begin{theorem}\label{theorem:classification for private variable normal form}
 Let $\protocol$ be a flat protocol where each $f_i$ has a private variable and which is in normal form. Then $\protocol$ is insecure if and only if a tracking strategy for $\protocol$ exists.
\end{theorem}

\begin{proof}
  \resetprooffacts
  Due to Theorem~\ref{theorem:tracking strategies work}, we know that if there is a tracking strategy for $\protocol$, then $\protocol$ is indeed insecure. Hence assume that there is no tracking strategy for $\protocol$, and indirectly assume that $\protocol$ is insecure. Let $\Pi$ be a corresponding successful adversary strategy. We introduce some notation for the proof:
   
 \begin{itemize}
  \item Since $\protocol$ is in normal form, we know that for each $i$, $\card{X_i}=1$, i.e., there is a single private variable for $f_i$. We denote this variable with $x_i$.
  \item In this proof, the index $i$ always ranges over the set $\set{1,\dots,n}$.
  \item A $V$-assignment is a function $g\colon V\rightarrow\set{0,1}$. Such an assignment is \emph{even}, if $\oplus_{y\in V}v(y)=0$ and \emph{odd}, if $\oplus_{y\in V}v(y)=1$. 
  \item For any $i$, a $V_i$-assignment is a function $g\colon V\cap\var{f_i}\rightarrow\set{0,1}$, i.e., a truth assignment to the non-private variables visible for $f_i$.
  \item We assume that $V=\set{y_1,\dots,y_m}$ for some $m\ge3$.
  \item We say that a local session $s$ is \emph{based} on a $V$-assignment $g$ ($V_i$-assignment $g$) if $s(y_k)=g(y_k)$ for all $y_k\in V$ (for all $y_k\in\var{f_i}\cap V$). The local session $s$ is even (odd) if it is based on an even (odd) assignment.
 \end{itemize}
 
 Similarly to the proof of Theorem~\ref{theorem:disjoint variables}, in the remainder of this proof we will only consider global sessions $(S,\sigma)$ satisfying the following conditions:
 
 \begin{enumerate}
  \item There are exactly $2^{\card{V}}$ local sessions, all of which are different assignments to the variables in $\var{\protocol}$,
  \item the set of queries at each $f_i$ is the set of all possible assignments $g\colon\var{f_i}\rightarrow\set{0,1}$,
  \item the schedule first performs all queries to $f_1$, then all questions to $f_2$, etc, the $g$-questions are scheduled last. For each service, the queries are scheduled in lexicographical order of the arguments.
 \end{enumerate}
  
  We call a global session satisfying the above a \emph{normal} global session. We will explicitly construct normal global sessions below. All normal global sessions look the same until the first $g$-query is performed. Therefore, similarly as in the proof for Theorem~\ref{theorem:disjoint variables}, $\Pi$ gives the same answer to each $f_i$ question in every normal global session. Since every $f_i$-question appears exactly once in each normal global session, $\Pi$ therefore is characterized by Boolean functions representing the answers to the $f_i$-questions and the output action performed at the end of the protocol run. (Clearly, the answers to the $g$-queries are not relevant.) We simply call these functions $f_1,\dots,f_n$, and for a local session $s$ of a normal global session simply write $f_i(s)$ for $f_i(\restr{s}{\var{f_i}})$. In the following, we identify a normal global session $S$ with the set of local sessions appearing in $S$, since the schedule is uniquely determined by this set and the conditions for normal global sessions. The adversary's view of a normal global session $S$ is completely determined by the set of $g$-queries performed in this session, i.e., the set of tuples $\set{(f_1(s),f_2(s),\dots,f_n(s))\ \vert\ s\in S}$.

  We now construct a global session $S_{\mathtext{even}}$ as follows: For each even $V$-assignment $g$, there are two local sessions $s^g_1$ and $s^g_2$ in $S_{\mathtext{even}}$ that are based on $g$. For each $i$, we have $s^g_1(x_i)=\overline{s^g_2(x_i)}$. We only fix this relationship now, and will fix the concrete values for the private variables later, depending on the adversary strategy $\Pi$. The global session $S_{\mathtext{odd}}$ is defined analogously, using the odd $V$-assignments. 
 
 \begin{prooffact}
  $S_{\mathtext{even}}$ and $S_{\mathtext{odd}}$ are normal global sessions.
 \end{prooffact}
 
 \begin{proof}
  By construction we have exactly $2^{\card{V}}$ local sessions, which represent pairwise different assignments to the variables. It remains to show that for each $f_i$ and each assignment $g\colon\var{f_i}\rightarrow\set{0,1}$, there is exactly one local session $s\in S_{\mathtext{even}}$ with $\restr{s}{\var{f_i}}=g$. Since $\card{\var{f_i}}=\card{V}$, it suffices to show that each two local sessions $s\neq s'\in S$ differ for at least one variable in $\var{f_i}$. Hence assume indirectly that  $\restr{s}{\var{f_i}}=\restr{s'}{\var{f_i}}$. In particular, $s(x_i)=x'(x_i)$. By construction of $S_{\mathtext{even}}$, $s$ and $s'$ are based on different $V$-assignments. Since $s$ and $s'$ agree on all $y_k\in V\cap\var{f_i}$, there is a unique $y_k\in V$ with $s(y_k)\neq s'(y_k)$. Hence $s$ and $s'$ cannot both be even. This is a contradiction, since $S_{\mathtext{even}}$ contains only local sessions based on even assignments. The proof for $S_{\mathtext{odd}}$ is analogous.
 \end{proof}
 
 For the variable $x_i$, the function $f_i$ is the only function that ``sees'' the value of $x_i$ in a local session. Therefore, $f_i$ has to ``report'' any information about the value of $x_i$ that the adversary wants to use. In particular, we are interested in $V_i$-assignments $g$ for which the value of $f_i$ differs depending on whether $g$ is extended by assigning $0$ or $1$ to $x_i$. These assignments are those where $f_i$ ``tries to'' keep track of the value of the variable $x_i$. We make this notion formal:

 \begin{notation}
  A $V$-assignment $g$ is $i$-flipping if $f_i(s)\neq f_i(s')$, where $\restr{s}V=\restr{s'}{V}=\restr{g}{\var{f_i}}$ and $s(x_i)\neq s'(x_i)$. We say that $g$ is all-flipping if $g$ is $i$-flipping for all $i$. A local session $s$ is $i$-flipping (all-flipping) if $\restr{s}{V}$ is $i$-flipping (all-flipping).
 \end{notation}
 
 Whether a $V$-assignment $g$ is $i$-flipping only depends on $\restr{v}{\var{f_i}}$. We therefore also apply the notion $i$-flipping to $V_i$-assignments with the obvious meaning.
  
 \begin{prooffact}\label{prooffact:adv strategy depends on private variables}\label{prooffact:one all-flipping even and odd assignment}
  \begin{enumerate}
   \item Let $s$ be the local session returned by $\Pi$ when a normal global session is run. Then $\restr{s}{V}$ is all-flipping.
   \item There are exactly two all-flipping $V$-assignments, one of which is even and one of which is odd. We denote these with $g_{\mathtext{even}}$ and $g_{\mathtext{odd}}$.
  \end{enumerate}
 \end{prooffact}
 
 \begin{proof}
  \begin{enumerate}
   \item Assume that this is false for the normal global session $S$. Then there is some $i$ such that $s$ is not $i$-flipping. Since $S$ is a normal global session, there is a session $s'\in S$ based on the same $V_i$-assignment as $s$ with $s'(x_i)=\overline{s(x_i)}$. Since $s$ is not $i$-flipping, we have that $f_i(s)=f_i(s')$. Let $S'$ be the global session obtained from $S$ by reversing the truth value of $x_i$ in both $s$ and $s'$. Since $f_i(s)=f_i(s')$, the value of $f_i$ remains unchanged for the sessions $s$ and $s'$.  Clearly, the result is a normal global session which does not contain the session $s$, and which is indistinguishable from $S$ for the adversary. Therefore, $\Pi$ again returns the session $s$, and thus fails on $S'$, a contradiction.
   \item Existence of the assignments directly follows from the above since $\Pi$ is a successful strategy and thus returns a correct session on both $S_{\mathtext{even}}$ and $S_{\mathtext{odd}}$. It remains to show that there is only one all-flipping $V$-assignment of each parity.
   
   We only cover the even case, the odd case is symmetric. Recall that $S_{\mathtext{even}}$ contains two sessions for every even $V$-assignment. Assume that there are $k$ even all-flipping assignments $g_1,\dots,g_k$ with $k\ge2$. We choose the values of the variables $x_i$ in the $S_{\mathtext{even}}$ such that for the two local sessions $s^1_1$ and $s^1_2$ based on $g_1$, we get the $g$-queries $(0,1,1,\dots,1)$ and $(1,0,0,\dots,0)$, and for $l\in\set{2,\dots,k}$, for the two local sessions $s^l_1$ and $s^l_2$ based on $g_l$, we get the $g$-queries $(1,1,1,\dots,1)$ and $(0,0,0,\dots,0)$. This is possible since $g_1,\dots,g_k$ are all-flipping.
  
   Since $\Pi$ is successful, $\Pi$ returns a local session $s\in S_{\mathtext{even}}$ when the global session $S_{\mathtext{even}}$ is run. Due to the first part, $s$ is all-flipping, and hence $s$ is based either on $g_1$ or on some $g_l$ with $l\ge2$.  We construct a session $S'_{\mathtext{even}}$ which contains none of the sessions from $S_{\mathtext{even}}$ based on $g_1$ or $g_l$. We do this by reversing the values of $x_1$ in these four local sessions. Clearly, the original four sessions are not contained in $S'_{\mathtext{even}}$. The $g$-queries from the affected sessions are as follows:

    \begin{itemize}
     \item For the two local sessions based on $g_1$, we get the queries $(1,1,1,\dots,1)$ and $(0,0,0,\dots,0)$, 
     \item For the two local sesisons based on $g_l$, we get the queries $(0,1,1,\dots,1)$ and $(1,0,0,\dots,0)$.
    \end{itemize}
  
    Hence the set of $g$-queries resulting from $S_{\mathtext{even}}$ and $S'_{\mathtext{even}}$ is identical. Since $S'_{\mathtext{even}}$ is obtained from $S_0$ by swapping the values of $x_i$ for the same $V$-assignment, $S'_{\mathtext{even}}$ is again a normal global session. Therefore, $S_{\mathtext{even}}$ and $S'_{\mathtext{even}}$ are indistinguishable to the adversary, and hence $\Pi$ returns the same local sessions $s$ when $S'_{\mathtext{even}}$ is run. Since this session is not present in $S'_{\mathtext{even}}$, we have a contradiction. 
  \end{enumerate}
 \end{proof}
 
 In the remainder of the proof, we only consider global sessions that satisfy the following criteria:
 
 \begin{itemize}
  \item For each $V$-assignment $g\colon V\rightarrow\set{0,1}$, there is exactly one local session $s$ based on $g$.
  \item Due to the above, for each $i$ and each $V_i$-assignment $g$, there are two local sessions $s^{i,g}_0$ and $s^{i,g}_1$ based on $g$. For these two sessions, we have that $s_0^{i,g}(x_i)=\overline{s_1^{i,g}(x_i)}$.
  \item The schedule follows the conditions for normal global sessions.
 \end{itemize}
 
 We call such global sessions \emph{complete normal sessions}, since every $V$-assignment appears, and it is easy to see that every such sessions is normal. For each $i$ and each $V_i$-assignment $g$, the local sessions $s^{i,g}_0$ and $s^{i,g}_1$ assign different values to the single variable $v_i\in V\setminus\var{f_i}$, and they also assign different values to the variable $x_i$. Therefore, either both of these sessions assign the same value to $x_i$ and $v_i$, or both of them assign different values to these variables.
 
 \begin{notation}
 Let $g$ be a $V_i$-assignment. As argued above, one of the two following cases occurs:
 \begin{itemize}
  \item Either $s^{i,g}_0(x_i)=s_0^{i,g}(v_i)$ and $s^{i,g}_1(x_i)=s_1^{i,g}(v_i)$, in this case we say that $x_i$ has positive polarity in $g$,
  \item or $s^{i,g}_0(x_i)=\overline{s^{i,g}_0(v_i)}$ and $s^{i,g}_1(x_i)=\overline{s^{i,g}_1(v_i)}$, in this case we say that $x_i$ has negative polarity in $g$.
 \end{itemize}
 \end{notation}
 
 We also say that $x_i$ has positive (negative) polarity in a $V$-assignment $g$, if $x_i$ has positive (negative) polarity in $g\cap\var{f_i}$.

 In the remainder of the proof, we will carefully change the polarities of variables for appropiate assignments to construct global sessions which are indistinguishable for the adversary, but do not contain the session that $\Pi$ reports. Hence we summarize the effect of reversing polarity of a variable in the following fact, which follows immediately from the definition of $i$-flipping: Reversing the polarity of $x_i$ for the $V$-assignment $v$ flips the value of $f_i(s)$ for both local sessions $s$ based on $v$.
 
 \begin{prooffact}\label{prooffact:effect of polarity change}
  Let $J\subseteq\set{1,\dots,n}$, let $s$ be a local session based on the $V$-assignment $g$, and let $s'$ be the local session obtained from $s$ by changing the polarities of all $x_j$ in $\restr{g}{\var{f_i}\cap V}$ with $j\in J$. Then
  
  $$f_i(s')=f_i(s)\oplus
  \begin{cases}
   1, & \mathtext{ if } i\in J\mathtext{ and }g\mathtext{ is }i\mathtext{-flipping}, \\
   0, & \mathtext{ otherwise.}
  \end{cases}$$
 \end{prooffact}

 Recall that due to Fact~\ref{prooffact:one all-flipping even and odd assignment}, there is exactly one all-flipping $V$-assignment of each parity, namely $g_{\mathtext{even}}$ and $g_{\mathtext{odd}}$. Clearly, $g_{\mathtext{even}}\neq g_{\mathtext{odd}}$, since their parity differs. Without loss of generality, we assume $g_{\mathtext{even}}(y_m)\neq g_{\mathtext{odd}}(y_m)$. We now make a case distinction, depending on whether there is another variable $y_k\in V\setminus\set{y_m}$ for which these assignments differ. These cases are significantly different, since if such a variable does not exist, then there are values $i$ (namely those for which $y_m\notin\var{f_i}$) for which $g_{\mathtext{even}}$ and $g_{\mathtext{odd}}$ agree on all variables in $V\cap\var{f_i}$, and hence from the input values for this $i$, sessions based on $g_{\mathtext{even}}$ and on $g_{\mathtext{odd}}$ cannot be distinguished. This makes the construction somewhat easier. This is not the case if $g_{\mathtext{even}}$ and $g_{\mathtext{odd}}$ differ for more than one variable in $V$.
  
 \paragraph{Case 1: $g_{\mathtext{odd}}(y_k)=g_{\mathtext{even}}(y_k)$ for all $k\in\set{1,\dots,m-1}$.} Then $g_{\mathtext{even}}$ and $g_{\mathtext{odd}}$ differ only in the value of $y_m$. In particular, the two local sessions $s_{\mathtext{even}}$ and $s_{\mathtext{odd}}$ based on $g_{\mathtext{even}}$ and $g_{\mathtext{odd}}$ agree on all $y_k$ for $k\in\set{1,\dots,m-1}$. Due to Fact~\ref{prooffact:effect of polarity change}, we can choose the polarities of the variables $x_i$ such that in the resulting global session $S$, we have the following:
 
 \begin{itemize}
  \item $f_i(s_{\mathtext{even}})=0$ for all $i$,
  \item in the local session $s_{\mathtext{odd}}$:
  \begin{itemize}
   \item for all $i$ with $y_m\in\var{f_i}$, let $f_i(s_{\mathtext{odd}})=0$.  This can be chosen independently from the above since  $g_{\mathtext{even}}$ and $g_{\mathtext{odd}}$ differ in a variable from $\var{f_i}\cap V$,
   \item for the remaining $i$ with $y_m\notin\var{f_i}$, we necessarily have $f_i(s_{\mathtext{odd}})=1$ since for these $i$, $g_{\mathtext{even}}$ and $g_{\mathtext{odd}}$ agree for the variables in $\var{f_i}\cap V$, hence 
   $s_{\mathtext{even}}(x_i)=\overline{s_{\mathtext{odd}}(x_i)}$. Since $g_{\mathtext{even}}$ and $g_{\mathtext{odd}}$ are all-flipping, this implies $f_i(s_{\mathtext{odd}})=\overline{f_i(s_{\mathtext{even}})}=1$.
  \end{itemize}
 \end{itemize}
 
 Hence the local session $s_{\mathtext{even}}$ results in the $g$-query $(0,0,\dots,0)$, and $s_{\mathtext{odd}}$ yields the query $(\underbrace{1,\dots,1}_{i\mathtext{ with } y_m\notin\var{f_i}},0,\dots,0)$. The polarities for the remaining sessions are irrelevant.
 
 Since $\Pi$ is successful, a local session $s$ from $S$ is returned when $S$ is run. Due to Fact~\ref{prooffact:adv strategy depends on private variables}, $s$ is all-flipping, and hence $s$ is based on $g_{\mathtext{even}}$ or on $g_{\mathtext{odd}}$, i.e., $s\in\set{s_{\mathtext{even}},s_{\mathtext{odd}}}$. We construct a complete normal session $S'$ that is indistinguishable from $S$ for the adversary, and which does not contain $s$.
 
 We construct $S'$ by reversing, in the assignment $g_{\mathtext{even}}$, the polarity of all $x_i$ with $y_m\notin\var{f_i}$. For such $i$, $g_{\mathtext{even}}$ and $g_{\mathtext{odd}}$ agree on all variables in $\var{f_i}\cap V$. Therefore, $S'$ contains neither $s_{\mathtext{even}}$ nor $s_{\mathtext{odd}}$, and hence does not contain $s$. We now show that $S'$ and $S$ are indistinguishable for the adversary, i.e., that they result in the same $g$-queries.
 
 Clearly, only the $g$-queries resulting from the local sessions based on $v_{\mathtext{even}}$ and $v_{\mathtext{odd}}$ can differ between $S$ and $S'$, and for both sessions, due to Fact~\ref{prooffact:effect of polarity change}, the values of the $f_i$-functions in these local sessions are reversed for all $i$ with $y_m\notin\var{f_i}$, since $g_{\mathtext{even}}$ and $g_{\mathtext{odd}}$ are $i$-flipping for all these $i$.
 Hence the local session in $S'$ based on $g_{\mathtext{even}}$ results in the $g$-query $(\underbrace{1,\dots,1}_{i\mathtext{ with } y_m\notin\var{f_i}},0,\dots,0)$, and the local session in $S'$ based on $g_{\mathtext{odd}}$ now results in the $g$-query $(0,0,\dots,0)$. The set of $g$-queries is therefore unchanged, and hence $S'$ is indistinguishable from $S$ for the adversary. Therefore, the adversary returns $s_1$ or $s_2$, fails on $S'$ as claimed.

 \paragraph{Case 2: there is some $k\in\set{1,\dots,m-1}$ with $v_{\mathtext{odd}}(y_k)=1$.} In this case $g_{\mathtext{even}}$ and $g_{\mathtext{odd}}$ differ in at least one variable from $\var{f_i}\cap V$ for all $i$. Hence for every $i$, we can set the polarities of $x_i$ in $g_{\mathtext{odd}}$ and in $g_{\mathtext{even}}$ independently.
  
 We again construct a complete global session as above by choosing appropiate polarities.
 
 In the following, let $\mathtext{pty}$ be either $\mathtext{even}$ or $\mathtext{odd}$.
 
 Let $s^{\mathtext{pty}}_1$ be the local session based on $v_{\mathtext{pty}}$, and let $s^{\mathtext{pty}}_2$ be the local session based on the assignment obtained from $v_{\mathtext{pty}}$ by reversing the value of $y_m$ (note that since $g_{\mathtext{even}}$ and $g_{\mathtext{odd}}$ differ in more than one variable, this assignment is different from both $g_{\mathtext{even}}$ and $g_{\mathtext{odd}}$, and in fact has a different parity than $\mathtext{pty}$).
 
 We consider the $g$-queries resulting from $s^{\mathtext{pty}}_1$ and $s^{\mathtext{pty}}_2$. These are the values $f_i(s^{\mathtext{pty}}_1)$ and $f_i(s^{\mathtext{pty}}_2)$ for all $i$. For all $i$, let $\alpha^{\mathtext{pty}}_i=f_i(s^{\mathtext{pty}}_1)$, and let $\beta^{\mathtext{pty}}_i=f_i(s^{\mathtext{pty}}_2)$. Since $s^{\mathtext{pty}}_1$ is based on $g_{\mathtext{pty}}$ which is all-flipping, we can choose the polarities of the $x_i$ for $g_{\mathtext{pty}}$ to achieve any value of $\alpha_i$. We fix the polarities as follows:
 
 \begin{itemize}
  \item For $i$ with $y_m\notin\var{f_i}$, we have that $s^{\mathtext{pty}}_1$ and $s^{\mathtext{pty}}_2$ agree on all variables in $V\cap\var{f_i}$. Hence $s^{\mathtext{pty}}_1(x_i)\neq s^{\mathtext{pty}}_2(x_i)$ must hold for such $i$. Since $g_{\mathtext{pty}}$ is all-flipping, we have $\beta^{\mathtext{pty}}_i=\overline{\alpha^{\mathtext{pty}}_i}$ for all such $i$.
  \item For $i$ with $y_m\in\var{f_i}$, since $s^{\mathtext{pty}}_1(y_m)\neq s^{\mathtext{pty}}_2(y_m)$, we can choose the polarities of $x_i$ in $s^{\mathtext{pty}}_1\cap V$ and in $s^{\mathtext{pty}}_2\cap V$ independently. We choose these polarities  such that $\alpha^{\mathtext{pty}}_i=\beta^{\mathtext{pty}}_i$. 
 \end{itemize}

 The remaining polarities are chosen arbitrarily. Let $S$ be the resulting global session, and let $s\in S$ be the session returned when $S$ is run. Due to Fact~\ref{prooffact:adv strategy depends on private variables}, $s$ is based on an all-flipping $V$-assignment, therefore $s\in\set{s^{\mathtext{even}}_1,s^{\mathtext{odd}}_1}$. Let $S'$ be the complete normal session obtained from $S$ by reversing the polarities of all $x_i$ with $y_m\notin\var{f_i}$ in $g_{\mathtext{even}}$ and in $g_{\mathtext{odd}}$. Clearly, both local sessions $s^{\mathtext{even}}_1$ and  $s^{\mathtext{odd}}_1$ do not occur in $S'$, and hence $s$ does not occur in $S'$. To show that $S$ and $S'$ are indistinguishable for the adversary, we show that the set of resulting $g$-queries is the same. Obviously, only the $g$-queries resulting from the local sessions based on the same $V$-assignments as the four sessions $s^{\mathtext{even}}_1$, $s^{\mathtext{even}}_2$, $s^{\mathtext{odd}}_1$ and $s^{\mathtext{odd}}_2$ are affected.
 
 Due to Fact~\ref{prooffact:effect of polarity change}, for the local sessions based on $\restr{s^{\mathtext{pty}}_1}{V}$ and $\restr{s^{\mathtext{pty}}_2}{V}$, the $g$-queries resulting from these local sessions in $S'$ are obtained from the values for the corresponding sessions in $S$, except that the values of the $f_i$ with $y_m\notin\var{f_i}$ are reversed.
 
 Therefore, the $g$-queries resulting from the sessions based on the same $V$-assignments as $s^{\mathtext{pty}_1}$ and $s^{\mathtext{pty}_2}$ are as follows:
 
 Instead of 
 $$(\underbrace{\alpha^{\mathtext{pty}}_1,\dots,\alpha^{\mathtext{pty}}_k,}_{i\mathtext{ with }y_m\notin\var{f_i}},\underbrace{\alpha^{\mathtext{pty}}_{k+1},\dots,\alpha^{\mathtext{pty}}_n}_{i\mathtext{ with }y_m\in\var{f_i}}) \mathtext{ and }(\underbrace{\overline{\alpha^{\mathtext{pty}}_1},\dots,\overline{\alpha^{\mathtext{pty}}_k},}_{i\mathtext{ with }y_m\notin\var{f_i}},\underbrace{\beta^{\mathtext{pty}}_{k+1},\dots,\beta^{\mathtext{pty}}_n}_{i\mathtext{ with }y_m\in\var{f_i}}),$$
 
 the resulting $g$-queries are now  
 $$(\underbrace{\overline{\alpha^{\mathtext{pty}}_1},\dots,\overline{\alpha^{\mathtext{pty}}_k},}_{i\mathtext{ with }y_m\notin\var{f_i}},\underbrace{\alpha^{\mathtext{pty}}_{k+1},\dots,\alpha^{\mathtext{pty}}_n}_{i\mathtext{ with }y_m\in\var{f_i}})\mathtext{ and }(\underbrace{\alpha^{\mathtext{pty}}_1,\dots,\alpha^{\mathtext{pty}}_k,}_{i\mathtext{ with }y_m\notin\var{f_i}},\underbrace{\beta^{\mathtext{pty}}_{k+1},\dots,\beta^{\mathtext{pty}}_n}_{i\mathtext{ with }y_m\in\var{f_i}}).$$
 
 Since for all $i$ with $y_m\in\var{f_i}$, we have that $\alpha^{\mathtext{pty}}_i=\beta^{\mathtext{pty}}_i$, these are the exact same sets of $g$-queries. Therefore, the two local sessions based on the same $V$-assignments as $s^{\mathtext{pty}}_1$ and $s^{\mathtext{pty}}_2$ lead to the same $g$-queries in both $S$ and $S'$. Since this is true for both parities, the set of $g$-queries in $S$ and in $S'$ is identical, and thus $S$ and $S'$ are indistinguishable for the adversary. It follows that $\Pi$ returns the local session $s$ when $S'$ is performed. Since $s\notin S'$, this shows that $\Pi$ fails.
 
 This completes the proof of Theorem~\ref{theorem:classification for private variable normal form}
\end{proof}

Using Theorem~\ref{theorem:classification for private variable normal form}, Theorem~\ref{theorem:each function has private variable} now follows rather easily:

\theoremeachfunctionhasprivatevariable*

\begin{proof}
 Due to Lemma~\ref{lemma:simple cases}, we can assume that none of the conditions stated in that lemma are satisfied, in particular, this means that each $f_i$ has a single private variable $x_i$. Since the existence of a tracking strategy implies insecurity due to Theorem~\ref{theorem:tracking strategies work}, we assume that there is no tracking strategy for $\protocol$ and prove that the protocol is secure. To do this, we indirectly assume that $\protocol$ is insecure.
 
 Due to Lemma~\ref{lemma:simple cases}, we can assume that there is no $i$ with $V\subseteq\var{f_i}$. Therefore, each $\var{f_i}$ contains at most $\card{V}-1$ of the variables in $V$. By adding edges from variable nodes to the nodes $f_i$, we can ensure that each $f_i$ sees exactly $\card{V}-1$ variables of $V$, since adding edges never makes an insecure protocol secure, the thus-obtained $\protocol$ is still insecure, and due to Lemma~\ref{lemma:tracking strategies for private variables}, there still is no tracking strategy for $\protocol$.
 
 If after these additions, there is a variable $v\in V$ such that $v\in\var{f_i}$ for all $i\in\set{1,\dots,n}$, then we remove this variable $v$ from $\protocol$ and the protocol remains insecure (again, this follows from Theorem~\ref{theorem:embedding} but can also be seen directly), and there still is no tracking strategy for $\protocol$ due to Lemma~\ref{lemma:tracking strategies for private variables}. Hence we can consecutively remove all variables $v$ appearing in all $\var{f_i}$. If for the thus-reduced protocol we have $\card{V}\leq 2$, the theorem follows from Lemma~\ref{lemma:simple cases}. Otherwise, we have obtained a flat protocol that is in normal form, insecure, but for which no tracking strategy exists. This is a contradiction to Theorem~\ref{theorem:classification for private variable normal form}.
\end{proof}

We therefore have seen that for flat protocols where every service-node has a private variable, the simple tracking strategy explained in the introduction is indeed the only successful adversary strategy in the sense that if this strategy fails, then all other strategies fail as well.

\section{Proof of Embedding Theorem~\ref{theorem:embedding}}

In this section, we prove Theorem~\ref{theorem:embedding}. We first discuss the insecurity-preserving transformations required for the proof. In addition to cloning, the following operations are used:

\begin{itemize}
 \item Introducing a \emph{bypass} for an edge $u\rightarrow v$ means connecting every predecessor of $u$ directly with $v$ and removing the edge $u\rightarrow v$.
 \item Removing a \emph{removable node} means removing a node $u$ that has no outgoing edges and for which there is a node $v\neq u$ such that each predecessor of $u$ is also a predecessor of $v$. 
 \item \emph{Splitting} an edge $u\rightarrow v$ means introducing a new node $w$, removing the edge $u\rightarrow v$ and introducing edges $u\rightarrow w$ and $w\rightarrow v$.
 \item \emph{Unsplitting} an edge means reversing the splitting operation.
 \item Restricting to a \emph{closed sub-protocol} means removing all nodes that are not part of an induced subgraph satisfying natural closure properties.
\end{itemize}

Figures~\ref{figure:protocol},~\ref{figure:bypass}, and~\ref{figure:cloning} visualize the effect of the bypass and cloning transformation: Figure~\ref{figure:protocol} contains an excerpt of a protocol. Figure~\ref{figure:bypass} shows the effect of bypassing the edge $u\rightarrow v$, and Figure~\ref{figure:cloning} shows the effect of cloning the node $v$. Splitting is visualized in Figures~\ref{figure:protocol 2} and~\ref{figure:splitting}; the former contains an excerpt of a protocol, the latter shows the effect of splitting on this excerpt. It is easy to see that the splitting operation does not affect the security of a protocol. The operation is interesting because it allows us to re-write protocols into a \emph{layered} form without affecting security.

\begin{minipage}{3.5cm}
 \begin{tikzpicture}
 \node at   (0,0) (u) {$u$};
 \node [above of = u] (v) {$v$};
 
 \node[below of = u] (tn)     {$t_n$};
 \node[left=-2mm of tn] (dots1) {$\dots$};
 \node[left=-2mm of dots1] (t1)     {$t_1$};
 
 \node[right of = u] (s1)     {$s_1$};
 \node[right=-2mm of s1] (dots2)     {$\dots$};
 \node[right=-2mm of dots2] (sm)     {$s_m$};
 
 \node[above of = v] (dots3) {$\dots$};
 \node[left=-2mm of dots3] (w1) {$w_1$};
 \node[right=-2mm of dots3] (wl) {$w_\ell$};
 
 \draw [->] (t1) edge [bend left] (u);
 \draw [->] (tn) edge [bend left] (u);

 \draw [->] (s1) edge [bend right] (v);
 \draw [->] (sm) edge [bend right] (v);
 
 \draw [->] (u) edge (v);
 
 \draw [->] (v) edge (w1);
 \draw [->] (v) edge (wl);
\end{tikzpicture}

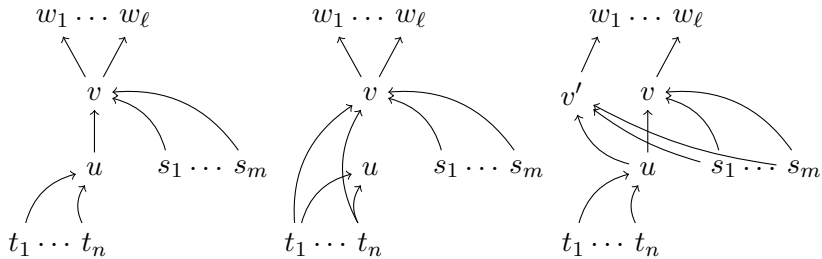
\captionof{figure}{Protocol}\label{figure:protocol}
\end{minipage}
\begin{minipage}{3.5cm}
\begin{tikzpicture}
 \node at   (0,0) (u) {$u$};
 \node [above of = u] (v) {$v$};
 
 \node[below of = u] (tn)     {$t_n$};
 \node[left=-2mm of tn] (dots1) {$\dots$};
 \node[left=-2mm of dots1] (t1)     {$t_1$};
 
 \node[right of = u] (s1)     {$s_1$};
 \node[right=-2mm of s1] (dots2)     {$\dots$};
 \node[right=-2mm of dots2] (sm)     {$s_m$};
 
 \node[above of = v] (dots3) {$\dots$};
 \node[left=-2mm of dots3] (w1) {$w_1$};
 \node[right=-2mm of dots3] (wl) {$w_\ell$};
 
 \draw [->] (t1) edge [bend left] (u);
 \draw [->] (tn) edge [bend left] (u);
 
 \draw [->] (t1) edge [bend left] (v);
 \draw [->] (tn) edge [bend left] (v);

 \draw [->] (s1) edge [bend right] (v);
 \draw [->] (sm) edge [bend right] (v);
 
 \draw [->] (v) edge (w1);
 \draw [->] (v) edge (wl);
\end{tikzpicture}
\captionof{figure}{Bypass}\label{figure:bypass}
\end{minipage}
\begin{minipage}{4cm}
 \begin{tikzpicture}
 \node at   (0,0) (u) {$u$};
 \node [above of = u] (v) {$v$};
 \node [left of = v]  (v') {$v'$};
 
 \node[below of = u] (tn)     {$t_n$};
 \node[left=-2mm of tn] (dots1) {$\dots$};
 \node[left=-2mm of dots1] (t1)     {$t_1$};
 
 \node[right of = u] (s1)     {$s_1$};
 \node[right=-2mm of s1] (dots2)     {$\dots$};
 \node[right=-2mm of dots2] (sm)     {$s_m$};
 
 \node[above of = v] (dots3) {$\dots$};
 \node[left=-2mm of dots3] (w1) {$w_1$};
 \node[right=-2mm of dots3] (wl) {$w_\ell$};
 
 \draw [->] (t1) edge [bend left] (u);
 \draw [->] (tn) edge [bend left] (u);

 \draw [->] (s1) edge [bend right] (v);
 \draw [->] (sm) edge [bend right] (v);

 \draw [->] (s1) edge [bend left=10] (v');
 \draw [->] (sm) edge [bend left=10] (v');

 \draw [->] (u) edge (v);
 \draw [->] (u) edge [bend left] (v');
 
 \draw [->] (v') edge (w1);
 \draw [->] (v) edge (wl);
\end{tikzpicture}
\captionof{figure}{Cloning}\label{figure:cloning}
\end{minipage}

We now give formal defintions of these treansformations and show that each of them preserves insecurity of a protocol. We then use these results to prove the theorem.

\begin{definition}
 Let $\protocol$ be a protocol, and let $u$, $v$ be non-variable nodes in $\protocol$.
 \begin{enumerate}
  \item if $u\rightarrow v$ is an edge, then the \emph{$u$-$v$-bypass} of $\protocol$ is obtained by the following operations:
    \begin{itemize}
      \item for each $w$ such that $(w,u)$ is an edge in $\protocol$, add an edge $(w,v)$ to $\protocol$,
      \item remove the edge $(u,v)$.
    \end{itemize}
  \item If $u$ has no outgoing edges in $\protocol$, then $u$ is \emph{removable}, if there is a node $w\neq u$ such that all predecessors of $u$ are also predecessors of $w$.
  \item Let $S$ be a subset of the successor nodes of $v$ in $\protocol$. Then \emph{$S$-cloning} $v$ results in a protocol obtained from $\protocol$ as follows:
  \begin{itemize}
   \item introduce a new node $v'$,
   \item for all edges $(w,v)$ in $\protocol$, introduce an edge $(w,v')$,
   \item for each $s\in S$, replace the edge $(v,s)$ with $(v',s)$.
  \end{itemize}
 \end{enumerate}
\end{definition}

``Cloning'' a variable node of an insecure protocol can result in a secure one, as the adversary now is required to take more variables into account. A simple example is the flat protocol $g(f_1(x,x),f_2(x,x))$, which is clearly insecure, but by repeatedly cloning $x$, we obtain the protocol $g(f_1(x,x'),f_2(x'',x'''))$, which is secure (see Example~\ref{example:simple secure protocol}, which discusses the same protocol with different variable names). Therefore, we only consider cloning non-variable nodes.

\begin{definition}
 Let $\protocol$ be a protocol, and let $(u,v)$ be an edge in $\protocol$. Then \emph{splitting $(u,v)$} results in the following protocol:
 \begin{itemize}
   \item remove the edge $(u,v)$,
   \item add a new node $w$,
   \item add edges $(u,w)$ and $(w,v)$.
 \end{itemize}
\end{definition}

\begin{wrapfigure}[12]{l}{6.5cm}
\begin{minipage}{3cm}
\begin{tikzpicture}
 \node at   (0,0) (u) {$u$};
 \node [above=10mm of u] (v) {$v$};
 \draw [->] (u) edge (v);

 \node [below of = u] (dots1) {$\dots$};
 \node [left=-2mm of dots1] (s1) {$s_1$};
 \node [right=-2mm of dots1] (sn) {$s_n$};
 
 \draw[->] (s1) edge [bend left=10] (u);
 \draw[->] (sn) edge [bend right=10] (u);
 
 \node [above of = v] (dots2) {$\dots$};
 \node [left=-2mm of dots2] (t1) {$t_1$};
 \node [right=-2mm of dots2] (tm) {$t_m$};
 
 \draw[->] (v) edge [bend left=10] (t1);
 \draw[->] (v) edge [bend right=10] (tm);
\end{tikzpicture}
\captionof{figure}{Protocol}
\label{figure:protocol 2}
\end{minipage}
\begin{minipage}{3cm}
\begin{tikzpicture}
 \node at   (0,0) (u) {$u$};
 \node [above=10mm of u] (v) {$v$};
 \node [above=2.5mm of u] (w) {$w$};
 \draw [->] (u) edge (w);
 \draw [->] (w) edge (v);
 
 \node [below of = u] (dots1) {$\dots$};
 \node [left=-2mm of dots1] (s1) {$s_1$};
 \node [right=-2mm of dots1] (sn) {$s_n$};
 
 \draw[->] (s1) edge [bend left=10] (u);
 \draw[->] (sn) edge [bend right=10] (u);
 
 \node [above of = v] (dots2) {$\dots$};
 \node [left=-2mm of dots2] (t1) {$t_1$};
 \node [right=-2mm of dots2] (tm) {$t_m$};
 
 \draw[->] (v) edge [bend left=10] (t1);
 \draw[->] (v) edge [bend right=10] (tm);
\end{tikzpicture}
\captionof{figure}{Splitting}
\label{figure:splitting}
\end{minipage}
\end{wrapfigure}

Intuitively, splitting $(u,v)$ introduces an additional node in the protocol that can be used by the adversary to simply copy its input. 

We show that all of the above-introduced operations preserve \emph{insecurity} of a protocol: If we apply one of the operations to an insecure protcol, then the resulting protocol is insecure as well. 

\begin{lemma}\label{lemma:bypass removable cloning}
 Let $\protocol$ be an insecure protocol, and let $\protocol'$ be obtained from $\protocol$ as
 \begin{enumerate}
  \item\label{enum:lemma:bypass removable cloning:cloning} the result of cloning a node $v$, or
  \item\label{enum:lemma:bypass removable cloning:bypass} the $u$-$v$-bypass of $\protocol$ for some nodes $u$ and $v$ of $\protocol$, or
  \item\label{enum:lemma:bypass removable cloning:removable node} the result of removing a single removable node and all of its incoming edges from $\protocol$.
 \end{enumerate}
 Then $\protocol'$ is insecure as well.
\end{lemma}

Note that we cannot remove all removable nodes of a protocol in one step, since there may be different nodes $u$ and $v$ with the same set of predecessors, but we cannot remove both of them without rendering the adversary's strategy defect.

\begin{proof}
  Note that in all three cases, the set of variables of $\protocol$ is unchanged. In particular, an assignment $I$ is a local session for $\protocol$ if and only if $I$ is a local session for $\protocol'$. Hence, let $S$ be a multiset of local sessions for $\protocol$ (or, equivalently, for $\protocol'$).
  \begin{enumerate}
    \item Since $\protocol$ is insecure, there is a successful strategy $\Pi$ against $\protocol$. We construct a strategy $\Pi'$ against $\protocol'$ that essentially simulates $\Pi$. The strategy $\Pi'$ proceeds as follows: During the run of the protocol, $\Pi$ keeps a list of entries of the form $(i,I_v,r)$, where $i\in\mathbb N$, $I_v$ is an assignment to the inputs to $v$, and $r$ is a single bit. Such an entry states that when $\Pi$ was queried for the $i$-th time with the input values $I_v$ at the node $v$, the reply was $i$. (Recall that an adversary's strategy may be inconsistent in the sense that when replying to the same question twice, it may produce different answers.) Additionally, $\Pi'$ keeps track of the queries made at the nodes $v$ and $v'$.
 
    \begin{itemize}
      \item for each query at a node $w\notin\set{v,v'}$: perform the query in a simulation of the original protocol with strategy $\Pi$.
      \item for each query at a node $w\in\set{v,v'}$: Let $I_v$ be the assignment to the input values from the user question. Let the current query be the $i$-th time that the value $I_v$ was asked at the node $w$. Check whether an entry $(i,I_v,r)$ is present in the above-mentioned list. If so, reply to the query with $r$. Otherwise, perform a query of the node $v$ in the simulated protocol $\protocol$ with strategy $\Pi$, and let $r$ be the reply. Store the entry $(i,I_v,r)$ and reply with $r$.
    \end{itemize}
 
  Clearly, this results in a simulation of $\protocol$ with strategy $\Pi$ on a global session with the same multiset of local sessions as the actually running global session. Essentially, $\Pi'$ simulates the caching of the value for $v$ that is done when $\protocol$ is performed. Therefore, since $\Pi$ is a successful strategy, $\Pi$ will eventually produce a local session. The strategy $\Pi'$ simply returns the local session produced by $\Pi$ and is successful.

  \item We first modify $\protocol$ so that $v$ is the only successor of $u$ in $\protocol$: We introduce a clone $u'$ of $u$, and connect all successors of $u$ that are different from $v$ to $u'$ instead. This protocol, which we denote with $\protocol_{u'}$ remains insecure due to point~\ref{enum:lemma:bypass removable cloning:cloning} of the present Lemma. We now connect all predecessors of $u$ directly to $v$, remove the node $u$, and call the resulting protocol $\protocol'$. Clearly, $\protocol'$ is identical to the protocol obtained as the $u$-$v$-bypass of $\protocol$ (the node $u'$ of $\protocol'$ is exactly the node $u$ of the bypass). Since $\protocol_{u'}$ is insecure, there is a strategy $\Pi$ for $\protocol_{u'}$. Let $g_1,\dots,g_n$ be the predecessor nodes of $u$ in $\protocol_{u'}$ (recall that $u$ cannot be a variable).
  The adversary can simulate performing the strategy $\Pi$ for the protocol $\protocol_{u'}$ when the protocol $\protocol'$ is run as follows:

  \begin{itemize}
   \item Note that no queries to $u$ appear, since $\protocol'$ does not contain the node $u$.
   \item For each query of some node different from $v$: Simply reply to the request as the original strategy $\Pi$ does.
   \item For each query of $v$: Due to construction of $\protocol'$, the query to $v$ in $\protocol'$ includes all arguments to a query of $u$ in $\protocol_{u'}$, i.e., the values for $g_1,\dots,g_n$. Therefore, the adversary can use the given values for $g_1,\dots,g_n$ to simulate the query to $u$ preceding the $v$-query in $\protocol_{u'}$, and hence, when answering the actual query to $v$, has access to the value of $u$ computed by the strategy $\Pi$. 
  \end{itemize}
  
  Since $\Pi$ is successful against all possible schedules, in particular for all schedules in which each $u$-query is directly followed by the $v$-query of the corresponding local session, the simulated strategy $\Pi$ successfully returns a local session from $S$. Hence the adversary simply can output this session.

  \item Let $\Pi$ be a strategy for the original protocol $\protocol$. We show how the adversary can apply a modification of $\Pi$ for the protocol $\protocol'$. Let $u$ be the removed node, and let $w$ be a node such that each $u$-predecessor (in $\protocol$) is also a $w$-predecessor (in $\protocol$). In this case, each node different from $u$ receives the same number of queries, independently of whether protocol $\protocol$ or $\protocol'$ is run. The adversary can simulate $\Pi$ as follows: After each query to $w$, simulate the corresponding query to $u$ (as in the case above, the adversary can obtain all necessary input-values for the simulated $u$-query from the actual $w$-query). Since $u$ has no outgoing edges, the value returned by $\Pi$ for the $u$-query can be ignored. Since $\Pi$ is successful in particular for all scheduler that perform each $u$-query directly after the $w$-query of the corresponding local session, the simulation is correct and hence, as above, the adversary can output the local session eventually returned by $\Pi$.
 \end{enumerate}
\end{proof}

\begin{lemma}\label{lemma:splitting edge}
 Let $\protocol'$ be a protocol obtained from $\protocol$ by splitting an edge $(u,v)$. Then $\protocol$ is secure if and only if $\protocol'$ is secure.
\end{lemma}

\begin{proof}
 Since reversing the splitting operation can be seen as applying a bypass transformation and removing a removable node, we know that $\protocol$ can be obtained from $\protocol'$ by these two operations. Therefore, due to Lemma~\ref{lemma:bypass removable cloning}, if $\protocol'$ is insecure, then $\protocol$ is insecure as well. 
 
 For the converse, assume that $\protocol$ is insecure, and let $w$ be the node introduced in the process of splitting $(u,v)$. Then clearly a successful strategy $\Pi$ for $\protocol$ can be applied for $\protocol'$ as follows:
 \begin{itemize}
  \item Answer every $w$-request with its input (i.e., the value for $u$).
  \item Due to the structure of $\protocol$, the $v$-euqry for each local session is performed after the $w$-query of the same local session. Therefore, the result of the $w$-query---due to the point above, this is simply the value for $v$---is then available for processing the $v$-query of the same local session, and hence the strategy $\Pi$ can be performed.
 \end{itemize}
\end{proof}

The final ``basic operation'' we discuss is essentially a subset condition: Any ``well-formed'' subset of an insecure protocol is insecure, provided that the subset contains all the nodes required to keep track of the variables relevant for the ``subset protocol.'' 

\begin{definition}
 A subset $\protocol'$ of a protocol $\protocol$ is a \emph{closed sub-protocol of $\protocol$} if all predecessors of nodes in $\protocol'$ are elements of $\protocol'$, and there is a node $u\in\protocol'$ that has no successors in $\protocol'$ and such that each path starting in $\protocol'$ and ending in $\protocol\setminus\protocol'$ visits $u$.
\end{definition}

\begin{lemma}\label{lemma:closed sub-protocol}
 Let $\protocol$ be insecure, and let $\protocol'$ be a closed sub-protocol of $\protocol$. Then $\protocol'$ is insecure as well.
\end{lemma}

\begin{proof}
 By definition of closed sub-protocols, every variable from $\var{\protocol'}$ is only connected to nodes in $\protocol\setminus\protocol'$ via an outgoing edge of $u$. Since $\protocol$ is insecure, let $\Pi$ be a strategy for $\protocol$. We construct a strategy $\Pi'$ for $\protocol'$ as follows:
 
 \begin{itemize}
  \item let all user sessions for the protocol $\protocol'$ complete. Use the strategy $\Pi$ to determine the answers to the queries. Since $\protocol'$ is predecessor-closed, these queries do not require any values from nodes in $\protocol\setminus\protocol'$.
  \item after all these sessions have completed: For each completed session, simulate a user session in the remainder of the protocol, i.e., $\protocol\setminus\protocol'$, using arbitrary values as user inputs, and the previously-determined return value of the query at $u$ when needed (i.e., for the simulation of $\Pi$ at the successor nodes of $u$).
  \item Since $\Pi$ is a correct strategy, $\Pi$ eventually returns a local session $I$ of $\protocol$. The restriction of $I$ to the variables in $\protocol'$ is a local session of $\protocol'$.
 \end{itemize}

\end{proof}

We now give the proof of the embedding theorem:

\theoremembedding*

\begin{proof}
 We use the results established in this section, to prove that if $\protocol$ is insecure, then so is $\protocol'$. We prove that $\protocol'$ can be obtained from $\protocol$ using the operations bypassing, removing removable nodes, cloning, splitting, reversing splitting, and taking closed sub-protocols. For this, we perform a number of transformations on $\protocol$. 
 
 \begin{enumerate}
  \item\label{embedding step:remove intermediate nodes} \emph{Remove intermediate nodes in $\protocol$.} While in $\protocol$, there is a path $\varphi(u)\rightsquigarrow v\rightsquigarrow\varphi(w)$ for $u,w\in\protocol'$ and $v\notin\varphi(\protocol')$, do the following:
     \begin{itemize}
        \item Let the path be $\varphi(u)\rightarrow v_1\rightarrow v_2\rightarrow\dots\rightarrow v_n=\varphi(w)$ with $v_i\notin\varphi(\protocol')$ for $i\in\set{1,\dots,n-1}$.
        \item Use the bypass operation to remove the edges $\varphi(u)\rightarrow v_1$ and $v_1\rightarrow v_2$, and introduce an edge $w\rightarrow v_2$ for each $w$ with $w\rightarrow v_1$.
        \item The remaining protocol remains insecure due to Lemma~\ref{lemma:bypass removable cloning}.\ref{enum:lemma:bypass removable cloning:bypass}. Further, since 
        \begin{enumerate}
          \item there is no pair $(u,v)$ for which there is a path $\varphi(u)\rightarrow\varphi(v)$ with only intermediate nodes from $\protocol\setminus\varphi(\protocol')$ after this change but not before, and
          \item we still have the property that every path leaving $\range(\protocol')$ does so via the $\varphi$-image of an output node of $\protocol'$,
          \item predecessor nodes of $\varphi$-images of $\var{\protocol'}$ are not affected,
        \end{enumerate}
        $\varphi$ remains a secure embedding.
        \item Continue the above operation until $n=1$, i.e., the path $\varphi(u)\rightarrow v_1\rightarrow v_2\rightarrow\dots\rightarrow v_n=\varphi(w)$ has been replaced with an edge $(\varphi(u),\varphi(w))$.
     \end{itemize}
     After this operation, there are no edges $(\varphi(u),v)$ in $\protocol$ anymore where $v\notin\varphi(\protocol')$ and $u$ has a successor in $\protocol'$, and $\varphi$ still is a secure embedding of $\protocol'$ into $\protocol$.

  \item \emph{Make $\varphi$ injective.} While there are nodes $u\neq v\in\protocol'$ with $\varphi(u)=\varphi(v)$, do the following, starting with the image of the root of $\protocol'$ (i.e., we can inductively assume that $\varphi$ is injective on the successors of $u$ and $v$, and hence $\varphi^{-1}(u)$ and $\varphi^{-1}(v)$ are well-defined):
    \begin{itemize}
       \item Let $w=\varphi(u)=\varphi(v)$
       \item Introduce a new node $w'$ into $\protocol$
       \item For each edge $(s,w)$ in $\protocol$, introduce an edge $(s,w')$
       \item For each edge $(w,t)$ such that there is no edge $v\rightarrow\varphi^{-1}(t)$ in $\protocol'$, remove the edge $(w,t)$ from $\protocol$ and add the edge $(w',t)$
       \item Redefine $\varphi(u)$ as $\varphi(u)=w'$.
    \end{itemize}
    Since this transformation is exactly the cloning operation, it follows from Lemma~\ref{lemma:bypass removable cloning}.\ref{enum:lemma:bypass removable cloning:cloning} that the resulting protocol is still insecure. Clearly, after the transformation, the modified $\varphi$ is still an embedding, and still there are no edges $(\varphi(u),v)$ in $\protocol$ where $v\notin\varphi(\protocol')$ and $u$ has a successor in $\protocol'$, and $\varphi$ still is a secure embedding of $\protocol'$ into $\protocol$.
   \item \emph{Remove irrelevant variables} For each $x\in\vars(\protocol)\setminus\range(\chi)$, remove the variable $x$ completely from $\protocol$. The protocol remains insecure, since the adversary can always simulate a protocol run with the variable $x$ still present by using the value $0$ for the input $x$. In addition, remove all nodes from $\protocol\setminus\varphi(\protocol')$ that have no ancestor nodes that are a variable. Clearly, $\protocol$ remains insecure after this transformation since these nodes cannot help the adversary's strategy (consider a schedule since all these nodes are queries before any other).
   \item \emph{Transform the inputs of the copy of $\protocol'$ to input nodes of $\protocol$.} For each $x\in\vars(\protocol')$, all variables $y$ from $\protocol$ such that $y\neq\chi(x)$ and there is a path from $y$ to $\varphi(x)$ have been removed in the step above. (Note that for each $x\in\range(\chi)$, there can be only one $w\in\vars(\protocol')$ such that $x\rightsquigarrow\varphi(w)$ is a path in $\protocol$.) After this, the path from $\chi(x)$ to $\varphi(x)$ in $\protocol$ consists of nodes with in-degree $1$, remove all of these nodes. The protocol remains insecure due to Lemma~\ref{lemma:splitting edge}. Now use Lemma~\ref{lemma:bypass removable cloning}.\ref{enum:lemma:bypass removable cloning:bypass} to connect $\chi(x)$ to all successors of $\varphi(x)$, remove all outgoing connections of $\varphi(x)$, remove the node $\varphi(x)$ using Lemma~\ref{lemma:bypass removable cloning}.\ref{enum:lemma:bypass removable cloning:removable node} and redefine $\varphi(x)=\chi(x)$. Clearly, the resulting protocol is still insecure and $\varphi$ remains an embedding. After this transformation, we have that if $u\rightarrow\varphi(v)$ is an edge in $\protocol$, then $u\in\varphi(\protocol')$, i.e., $\varphi(\protocol')$ is predecessor-closed: Each remaining node in $\protocol\setminus\varphi(\protocol')$ has an ancestor $\chi(x)$ for some $x\in\var{\protocol'}$. Hence if $u\rightarrow\varphi(v)$ is an edge, then we have that $\chi(x)=\varphi(x)\rightsquigarrow u\rightarrow (v)$. Such a path with $u\notin\varphi(\protocol')$ does not exist in $\protocol$ anymore after application of step~\ref{embedding step:remove intermediate nodes}.
 \end{enumerate}
 
 Since all the above steps preserve the insecurity of $\protocol$ and the fact that $\varphi$ is a secure embedding of $\protocol'$ into $\protocol$, we therefore can without loss of generality assume that $\protocol$ is already the result of the above steps. In particular, this implies:
     
 \begin{itemize}
  \item there are no paths $\varphi(u)\rightsquigarrow v\rightsquigarrow(w)$ in $\protocol$ with $v\notin\varphi(\protocol')$,
  \item $\varphi$ is injective,
  \item for each $x\in\vars(\protocol')$, we have that $\chi(x)=\varphi(x)\in\vars(\protocol)$.
  \item no variable $x\in\vars(\protocol)\setminus\range(\chi)$ is connected to a node in $\varphi(\protocol')$.
  \item $\varphi(\protocol')$ is predecessor-closed.
 \end{itemize}
 
 We now show that $\varphi(\protocol')$ is a closed sub-protocol. 

 \begin{itemize}
  \item By definition, $\varphi(\protocol')$ is an induced subgraph, and by the above, $\varphi(\protocol')$ is predecessor-closed.
  \item We show that every path from $\varphi(\protocol')$ to a node from $\protocol\setminus\varphi(\protocol')$ visits $\varphi(r_{\protocol'})$, where $r_{\protocol'}$ is the root of $\protocol'$. Hence let $\varphi(u)\rightsquigarrow v$ be a path in $\protocol$, where $v\notin\varphi(\protocol')$, and $u\neq r_{\protocol'}$. Since $\varphi$ is an embedding, there is some $w\in\protocol'$ such that $v\rightsquigarrow\varphi(w)$ is a path in $\protocol$. This is a contradiction, since due to the above, there are no such paths in $\protocol$ anymore after the transformation.
 \end{itemize}
 
 Therefore, $\varphi(\protocol')$ is indeed a closed sub-protocol. Since $\protocol$ is insecure, it follows from Lemma~\ref{lemma:closed sub-protocol} that $\varphi(\protocol')$ is insecure as well. 
 
 Since all intermediate nodes on paths $\varphi(u)\rightsquigarrow\varphi(v)$ that are no elements from $\varphi(\protocol')$ have been removed, and for all $u\rightarrow v$, such a path $\varphi(u)\rightsquigarrow\varphi(v)$ exists in $\protocol'$, it follows that if $\varphi(u)\rightarrow\varphi(v)$ is an edge in $\protocol$, then $u\rightarrow v$ is an edge in $\protocol'$. We can without loss of generality assume that the other direction is true as well, since adding edges only makes a protocol more insecure. Since $\varphi$ is injective, this implies that $\varphi$ is in fact an isomorphism. Therefore, $\protocol'\approx\varphi(\protocol')$ is insecure as well as claimed.
\end{proof}

\section{Security Proofs for Deep Protocols}

We now give the proofs for the corollaries in Section~\ref{sect:embedding application}. Both of these results follow easily by applying the embedding technique introduced in Section~\ref{sect:deep:embedding} to the results obtained for flat protocols in Section~\ref{sect:two cases flat secure protocols}.

\corollarydisjointvariablesgeneralization*

\begin{proof}
 This result follows directly from Theorem~\ref{theorem:disjoint variables} and Theorem~\ref{theorem:embedding}: Consider the protocol $\protocol'=g(f_1(x^1_1,\dots,x^1_{k_1},\dots,f_n(x^n_1,\dots,x^n_{k_n})$, where $\set{x^i_1,\dots,x^i_{k_i}}$ is the set of variables $x$ such that $x\rightsquigarrow f_i$ is a path in $\protocol$. Then, due to Theorem~\ref{theorem:disjoint variables}, the protocol $\protocol'$ is secure. Clearly, the function $\varphi$ defined with $\varphi(x)=x$ for all $x\in\var{\protocol}$, $\varphi(f_i)=f_i$ for all relevant $i$, and $\varphi(g)=r_\protocol$ where $r_\protocol$ is the root of $\protocol$ (if $\protocol$ does not have a root, we add one connected to all nodes of $\protocol$ without an outgoing edge), with the function $\chi(x)=x$ for all $x\in\var{\protocol}$ constitutes a secure embedding of $\protocol'$ into $\protocol$. Therefore, Theorem~\ref{theorem:embedding} implies that $\protocol$ is secure.
\end{proof}

\corollaryprivatevariablesgeneralization*

\begin{proof}
 This result follows directly from Theorem~\ref{theorem:each function has private variable} and Theorem~\ref{theorem:embedding}: Let $L_i=\set{f_1,\dots,f_m}$, and for each $f_j\in L_i$, let $\vars(f_j)=\set{x^j_1,\dots,x^j_{k_j}}$. By theorem~\ref{theorem:each function has private variable}, the protocol $\protocol':=g(f_1(x^1_1,\dots,x^1_{k_1}),\dots,f_m(x^m_1,\dots,x^m_{k_m}))$ is secure. Due to Theorem~\ref{theorem:embedding}, it suffices to construct a secure embedding $\varphi$ of $\protocol'$ into $\protocol$. We define this embedding by $\varphi(x)=x$ for all $x\in\var{\protocol}$, $\varphi(f_j)=f_j$ for all $f_i\in L_i$, and $\varphi(g)=r_\protocol$ where $r_\protocol$ is the root of $\protocol$ (if $\protocol$ does not have a root, we add one just as in the proof of Corollary~\ref{corollary:disjoint variables generalization}). We define the function $\chi$ as $\chi(x)=x$ for each $x\in\var{\protocol}$. Clearly, this constitutes an embedding, and hence $\protocol$ is secure as claimed.
\end{proof}

\end{appendix}

\end{document}